\documentclass{article}

\usepackage{fullpage}
\usepackage{authblk}
\usepackage[hidelinks]{hyperref}

\usepackage[american]{babel}	%English hyphenation etc.
\usepackage[utf8]{inputenc}		%Input encoding
\usepackage[T1]{fontenc}		%Font encoding
\usepackage[english=american]{csquotes}	%Quotation marks ' '
\usepackage{natbib}
\usepackage{graphicx}			%Include graphics
\usepackage{epstopdf}
\usepackage{subfig}				%Subfigures
\usepackage{marvosym}			%Special symbols
\usepackage{amsmath}			%Mathmode
\usepackage{amssymb}
\usepackage{amsthm}
\usepackage{mathtools}
\usepackage{color}
\usepackage[linesnumbered,ruled]{algorithm2e}
\usepackage[normalem]{ulem}
\usepackage{nicefrac}
\usepackage{soul}
\usepackage{amsfonts}    	
\usepackage{xcolor}
\usepackage{csquotes}
\usepackage{float}
\usepackage{url}
\usepackage[format=plain,labelfont=bf]{caption}
\usepackage{capt-of}
\usepackage{lmodern}

\newcommand\blfootnote[1]{%
  \begingroup
  \renewcommand\thefootnote{}\footnote{#1}%
  \addtocounter{footnote}{-1}%
  \endgroup
}

\theoremstyle{definition}
\newtheorem{definition}{Definition}%[subsection]
\newtheorem{remark}{Remark}

\newtheorem{observation}{Observation}

\theoremstyle{plain}
\newtheorem{lemma}{Lemma}
\newtheorem{theorem}{Theorem}%[section]
\newtheorem{proposition}{Proposition}

\title{Unrooted non-binary tree-based phylogenetic networks}

\author[1, $\ast$]{Mareike Fischer}
\author[1,3]{Lina Herbst}
\author[1]{Michelle Galla}
\author[2]{Yangjing Long}
\author[1]{Kristina Wicke}

\affil[1]{Institute of Mathematics and Computer Science, University of Greifswald, Greifswald, Germany}
\affil[2]{School of Mathematics and Statistics, Central China Normal University, Wuhan, Hubei, China}
\affil[3]{Transmission, Infection, Diversification \& Evolution Group, Max Planck Institute for the Science of Human History, Jena, Germany}

\begin{document}
\maketitle

\begin{abstract} 
Phylogenetic networks are a generalization of phylogenetic trees allowing for the representation of non-treelike evolutionary events such as hybridization. Typically, such networks have been analyzed based on their `level', i.e.  based on the complexity of their 2-edge-connected components. However, recently the question of how `treelike' a phylogenetic network is has become the center of attention in various studies. This led to the introduction of \emph{tree-based networks}, i.e. networks that can be constructed from a phylogenetic tree, called the \emph{base tree}, by adding additional edges. While the concept of tree-basedness was originally introduced for rooted phylogenetic networks, it has recently also been considered for unrooted networks. In the present study, we compare and contrast findings obtained for unrooted \emph{binary} tree-based networks to unrooted \emph{non-binary} networks. In particular, while it is known that up to level 4 all unrooted binary networks are tree-based, we show that in the case of non-binary networks, this result only holds up to level 3.
\end{abstract}


\blfootnote{$^\ast$Corresponding author\\ \textit{Email address:} \texttt{email@mareikefischer.de} (Mareike Fischer)}

\section{Introduction}
Traditionally, phylogenetic trees have been used in order to represent the evolutionary history of sets of species. However, the evolution of species is not always tree-like, in particular if they have been subject to reticulation events like hybridization or horizontal gene transfer. Thus, phylogenetic networks have come to the fore as a mathematical generalization of phylogenetic trees, allowing for the representation of non-tree-like evolutionary events. 

Mathematically, unrooted phylogenetic networks are connected, simple graphs containing degree-1 vertices (representing present-day species) that are not necessarily acyclic, as reticulation events may lead to 2-edge-connected components in the graph. In this regard, \citet{Choy2005} introduced the concept of \emph{level-$k$ networks}. These are networks in which in each 2-edge-connected component the number of edges which need to be removed to turn them into trees is at most $k$. In this sense, the level measures the complexity of a network. So the smaller the level, the more `treelike' the network. 

However, while a network can contain various trees, not all networks are suitable to explain the evolution of present-day species to the same extent. This is due to the fact that some networks contain no \emph{support tree} (\cite{Francis2015}). A support tree is a spanning tree that has the same leaf set as the network, i.e. it represents the evolution of the same present-day species as the given network. Biologically, these networks are most relevant, because the leaves typically represent the species for which one has data (e.g. DNA sequences) and on which the reconstruction of evolution is based.

So given a phylogenetic network, it is of high interest to study its `treelikeness', and in terms of support trees, this reduces to the question whether the network is merely a tree with additional edges. While \cite{Francis2015} introduced the concept of tree-basedness for rooted binary phylogenetic networks, recently \cite{Francis2018} extended it to unrooted binary networks, \cite{Jetten2018} to rooted non-binary networks and \cite{Hendriksen2018} to unrooted non-binary networks. 

In the present manuscript, we focus on unrooted non-binary networks and link and contrast our findings to results recently obtained for unrooted binary networks by \citet{Francis2018}. In particular, we show that while up to level 4 all unrooted binary networks are tree-based (Theorem 1 in \citet{Francis2018}, in case of non-binary networks, this result only holds up to level 3 (Theorem \ref{level_non-binary} of the present manuscript). More precisely, there are non-binary level-4 networks that are not tree-based. This provides an answer to Question 5.3 posed in \citet{Hendriksen2018}, asking whether there are networks of level less than 5 that are not tree-based.\footnote{Note that our definition of tree-based networks corresponds to the definition of \emph{loosely tree-based} networks in \citet{Hendriksen2018} (see Section 2.1 in the present manuscript), so the question is posed there in a slightly different way.}

Our analysis of the relationship between the level of a non-binary network and its tree-basedness relies on certain properties of the graph underlying the network (e.g. the number of vertices in the `blobs' of the network). A similar analysis can also be used to show that all unrooted \emph{binary} networks up to level 4 are tree-based, i.e. to re-prove Theorem 1 of \citet{Francis2018}. We provide the details of this alternative proof in the Appendix. Importantly, this allows us to show that there are precisely 2 `minimal' (in the number of vertices) proper unrooted binary phylogenetic level-5 networks that are not tree-based; i.e. such networks are very rare.

The remainder of this manuscript is organized as follows. In Section \ref{sec_preliminaries}, we introduce some basic phylogenetic and graph-theoretical concepts and terminology.
In Section \ref{sec_basic} we then consider several basic results obtained for unrooted binary networks in \citet{Francis2018} and generalize them to non-binary networks. We then use these results to state the main theorem of this paper saying that all unrooted non-binary networks up to level 3 are tree-based in Section \ref{sec_non-binary}. In Section \ref{sec_binary} we re-visit the corresponding result for unrooted binary networks obtained by \citet{Francis2018} (Theorem 1 therein) and show that there are exactly two unrooted binary level-5 networks with 12 vertices that are not tree-based. Moreover, we provide an alternative proof of Theorem 1 in \citet{Francis2018} in the Appendix.

\section{Preliminaries} \label{sec_preliminaries}
\subsection{Phylogenetic Concepts}
\subsubsection*{Phylogenetic Networks}
Throughout this manuscript we assume that $X$ is a finite set (e.g. of taxa or species) with $\vert X \vert \geq 1$.
A \emph{non-binary unrooted phylogenetic network} $N$ (on $X$) is a connected, simple graph $G=(V,E)$ with $X \subseteq V$ and no vertices of degree 2, where the set of degree 1 vertices (referred to as the \emph{leaves} or \emph{taxa} of the network) is bijectively labeled by and thus identified with $X$. 
An unrooted phylogenetic network is called \emph{unrooted binary} if every non-leaf vertex $u \in V \setminus X$ has degree 3.\footnote{Note that, somewhat counterintuitively, the class of binary phylogenetic networks is contained in the class of non-binary ones, as non-binary ones are simply more general. This definition is common in the phylogenetic literature.} In the following, we denote by $\mathring{E}$ the set of interior edges of $N$, i.e. those edges that are {\em not} incident to a leaf.

Moreover, note that an unrooted \emph{phylogenetic tree} is an unrooted phylogenetic network whose underlying graph structure is a tree.

\subsubsection*{Edge subdivision and vertex suppression}
In order to discuss the notion of tree-basedness, we now need to introduce the concepts of \emph{subdividing an edge} and \emph{suppressing a vertex}. Therefore, let $N$ be a phylogenetic network with some edge $e=\{u,v\}$. Then, we say that we \emph{subdivide} $e$ by deleting $e$, adding a new vertex $w$ and adding the edges $\{u,w\}$ and $\{w,v\}$. The new degree 2 vertex $w$ is sometimes also referred to as an \emph{attachment point}. Note that we often also refer to the vertex adjacent to a leaf $x$ as the attachment point of $x$, even if this vertex has degree higher than two.

In contrast to adding vertices to a network, given a degree 2 vertex $w$ with adjacent vertices $u$ and $v$, by \emph{suppressing $w$} we mean deleting $w$ and its two incident edges $\{u,w\}$ and $\{w,v\}$ and adding a new edge $\{u,v\}$. Note that the resulting graph can be a multigraph, i.e. a graph that may contain parallel edges and/or loops (cf. Figure \ref{fig_suppressing}). 

\begin{figure}[htbp] 
	\centering
	\includegraphics[scale=0.15]{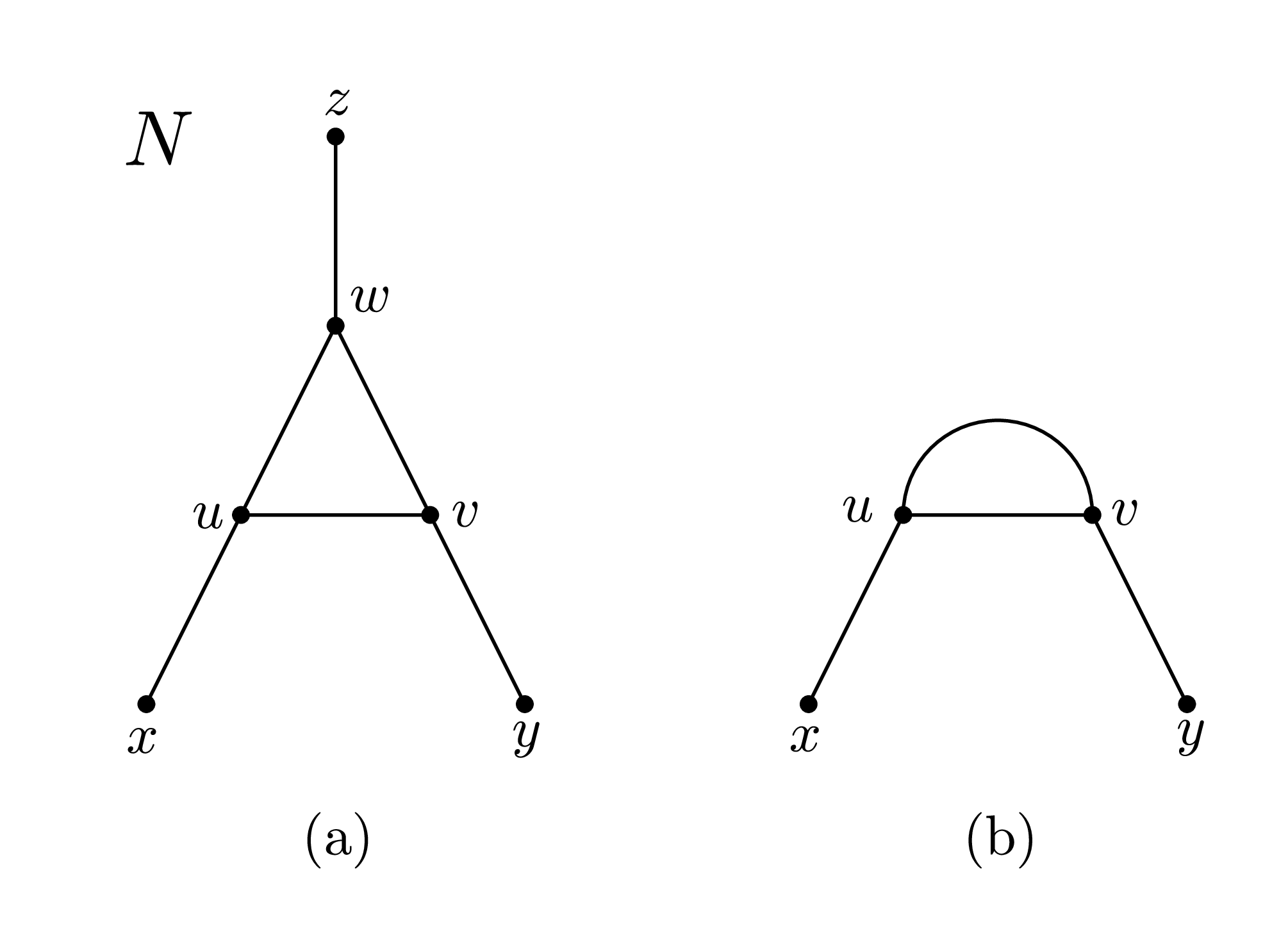}
	\caption{Unrooted phylogenetic network $N$ on three leaves (a). Deleting leaf $z$ and suppressing the resulting degree 2 vertex $w$, i.e. deleting $w$ and its incident edges $\{u,w\}$ and $\{w,v\}$ and adding a new edge $\{u,v\}$ results in a multigraph, because $u$ and $v$ are already connected by an edge (b).}
	\label{fig_suppressing}
\end{figure}

\subsubsection*{Tree-based networks}
Edge subdivisions and attachment points play a fundamental role in the concept of tree-based networks. 
When tree-basedness was first introduced for rooted binary phylogenetic networks by \citet{Francis2015}, tree-based networks were constructed from rooted binary phylogenetic trees by subdividing edges and adding new edges between such pairs of attachment points. However, tree-based networks (whether they are rooted or unrooted, binary or non-binary) can alternatively be defined as follows. For technical purposes (e.g. for Lemma \ref{N-x} in the Appendix), we first define tree-basedness for general (multi)graphs before defining it for phylogenetic networks (which are by definition \emph{simple} graphs).

\begin{definition} \label{def_tree-based}
Let $G=(V,E)$ be a connected (multi)graph and with leaf set $V^1$, i.e. $V^1=\{v \in V: \, deg(v) \leq 1\}$. $G$ is called \emph{tree-based} if there is a spanning tree $T=(V,E')$ in $G$ (with $E' \subseteq E$) whose leaf set is equal to $V^1$. $T$ is then called a \emph{support tree} (for $G$). 
Moreover, tree $T'$ which can be obtained from $T$ by suppressing potential degree 2 vertices is called a \emph{base tree} (for $G$). 
If $N$ is a phylogenetic network with leaf set $X$ and underlying graph $G$ and $G$ is tree-based, then we call $N$ a \emph{tree-based network} with support tree $T$ (and base tree $T'$).
\end{definition}

Note that the existence of a support tree $T$ for $G$ implies the existence of a base tree $T'$ for $G$.

Moreover, note that while we consider only one kind of tree-basedness for unrooted non-binary networks, one might want to distinguish between several forms of tree-based networks, depending on different ways of introducing additional edges (cf. \citet{Hendriksen2018}). For instance, \citet{Hendriksen2018} distinguished between strictly tree-based, tree-based and loosely tree-based networks, where the class of loosely tree-based networks is the most general one, which contains the other two. The latter class, i.e. the one that \citet{Hendriksen2018} calls the class of loosely tree-based, is precisely the class of networks we call tree-based in the present manuscript.

\subsubsection*{Cut edges/vertices, blobs, level-$k$ and proper networks} \label{blobs}
We will see in subsequent sections that it is often useful to decompose a phylogenetic network into simpler pieces, which can then be analyzed individually. Therefore, recall the following definitions from \citet{Gambette2012}. Let $N$ be an unrooted network. 
A \emph{cut edge}, or \emph{bridge}, of $N$ is an edge $e$ whose removal disconnects the graph, i.e. an edge $e$ such that $N-e$ is disconnected. Similarly, we call a vertex $v$ a \emph{cut vertex}, if deleting $v$ and all its incident edges disconnects the graph.

A cut edge is called \emph{trivial} if one of the connected components induced by the removal of the cut edge is a single vertex (which must necessarily be a leaf). Following \citet{Gambette2012}, we call $N$ a \emph{simple} network if all of its cut edges are trivial.
A \emph{blob} in a network is a maximal connected subgraph that has no cut edge. If a blob consists only of one vertex, we call the blob \emph{trivial}. Note that this implies that we can consider a network as a \enquote{tree-like structure} with blobs as vertices (cf. Figure \ref{blob}). The idea of \enquote{blobbed trees} has already been introduced for rooted phylogenetic networks in \citet{Gusfield2005} and we use it for unrooted ones in the following. Moreover, note that in a binary network it can be easily seen that a blob not only contains no cut edges, but it contains no cut vertices either (because in binary networks, every cut vertex is incident to a cut edge, and these are excluded from blobs, cf. Lemma \ref{binarycutvertex} in the Appendix). In contrast, in the non-binary setting, a blob may contain cut vertices. An example for such a blob can be seen in Figure \ref{blob}. 

\begin{figure}[htbp] 
	\centering
	\includegraphics[scale=0.3]{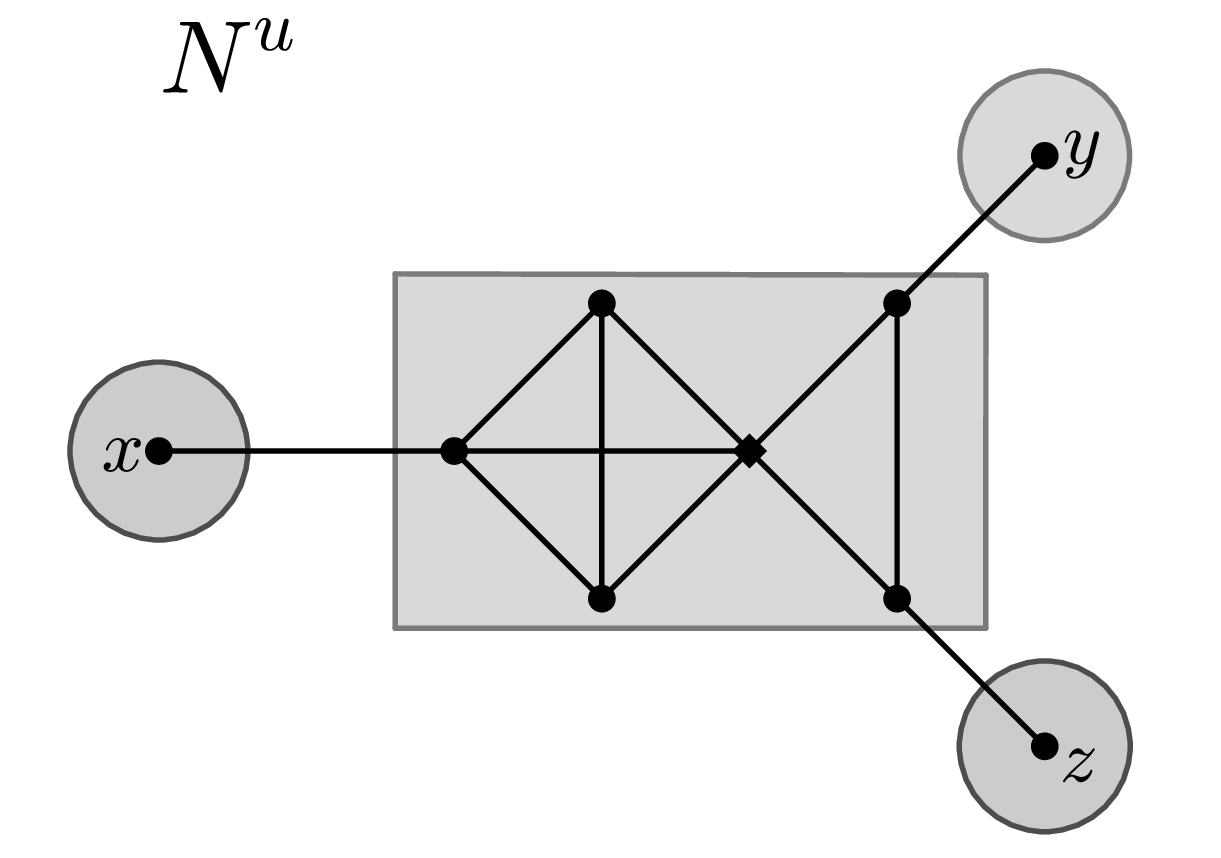}
	\caption{Unrooted non-binary phylogenetic network $N$ on taxon set $X=\{x,y,z\}$. The gray areas correspond to the blobs of $N$. Note that $N$ consists of three trivial blobs and one non-trivial blob and this non-trivial blob contains a cut vertex (depicted as a square vertex). Moreover, note that the cut edges and blobs in $N$ induce a \enquote{tree-like structure}, i.e. $N$ can be considered as a (graph-theoretical) tree with blobs as vertices.}
	\label{blob}
\end{figure}

Recall that a binary network $N$ on $X$ is called \emph{proper} (\citet{Francis2018}) if every cut edge induces a \emph{split} of $X$, i.e. a bipartition of $X$ into two non-empty subsets. Here, however, we call a network \emph{proper} if the removal of any cut edge or cut vertex present in the network leads to connected components containing at least one leaf each. 
Note that this definition of proper networks generalizes the one given in \citep{Francis2018} to the non-binary case. There, only cut edges were considered for proper networks. However, the alteration of the definition in the non-binary case is needed in order to exclude some networks which cannot be tree-based. In particular, we exclude non-tree networks in which all leaves are attached to the same interior vertex.

\begin{remark}  \label{all_leaves_one_vertex}
A network which is not a tree and for which $|V|>|X|\geq 1$ such that all leaves are attached to the same interior vertex cannot be tree-based. This is due to the fact that any spanning tree will induce additional leaves that are not part of $X$, i.e. no spanning tree of the network is a support tree (cf. Figure \ref{Fig_LeafAttach}). 
Moreover, note that a network with $\vert X \vert = 1$ and $\vert V \vert > \vert X \vert$ cannot be proper, because if $\vert X \vert = 1$ no cut edge or cut vertex can induce a partition of the taxon set. 
\end{remark}

\begin{figure}[htbp]
\centering
\includegraphics[scale=0.4]{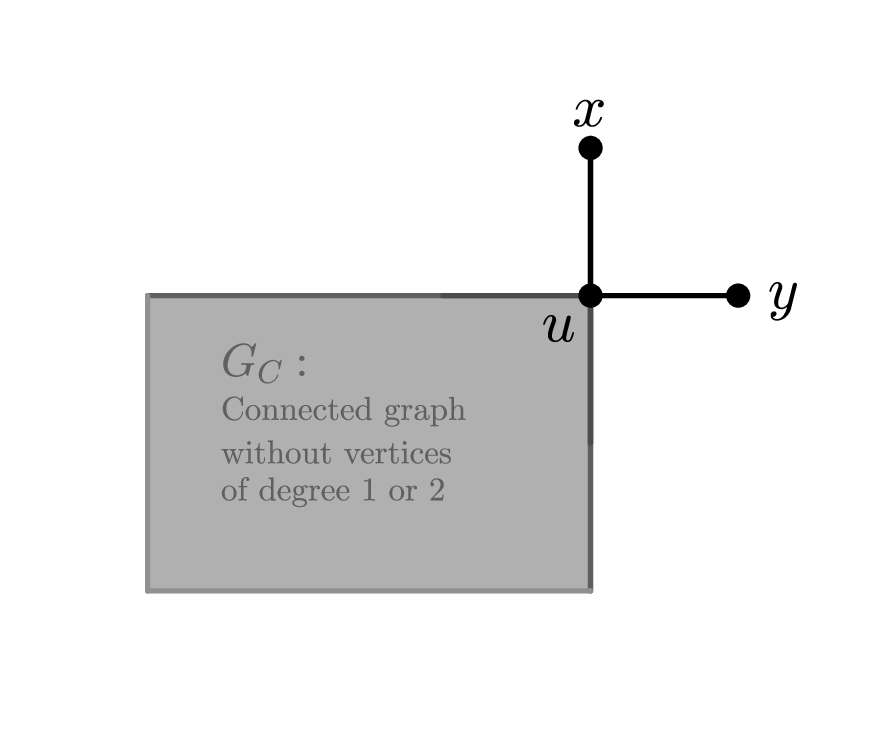}
\caption{Network $N$ with $\vert V \vert > \vert X \vert \geq 1$ such that all leaves, i.e. $x$ and $y$, are attached to the same interior vertex $u$. Any support tree for $N$ would have to cover all vertices of $G_C$, as well as $u$, $x$ and $y$. Here, $G_C$ is a connected graph without vertices of degree 1 or 2 (in particular $G_C$ is not a tree). Thus, any spanning tree would induce at least one additional leaf besides $x$ and $y$. In particular, it would not be a support tree and thus $N$ cannot be tree-based.}
\label{Fig_LeafAttach}
\end{figure}

Given a network $N$ on $X$ and an integer $k \geq 0$, we call $N$ a \emph{level-}$k$ network if at most $k$ edges have to be removed from each blob of $N$ to obtain a tree. 

Moreover, following \citet{Francis2018}, given a network $N$ and a blob $B$ in $N$, we define a simple network $B_{N}$ by taking the union of $B$ and all cut edges in $N$ incident with vertices in $B$, where the leaf set of $B_{N}$ is simply the set of end vertices of these cut edges that are not already a vertex in $B$.

\subsection{Graph-theoretical concepts} \label{graphtheo}
Besides the phylogenetic terminology, we need to introduce some basic concepts from classical graph theory before we can proceed with analyzing the tree-basedness of unrooted phylogenetic networks. 
In particular, we need the notion of \emph{cubic graphs} and \emph{Hamiltonian paths/cycles.}

A \emph{cubic graph} is a graph $G=(V,E)$ such that all vertices have degree 3.
We will now show that both cubic graphs and unrooted binary phylogenetic networks necessarily contain an even number of vertices, a property that will be useful later on (e.g. in the proof of Lemma \ref{12vertices}). 

\begin{proposition} \label{prop_cubic_even}
Let $G=(V,E)$ be a cubic graph. Then, $G$ contains an even number of vertices. Furthermore, if $N$ is an unrooted binary phylogenetic network with $n > 1$ leaves, $N$ also contains an even number of vertices.
\end{proposition}

\begin{proof}
The first part of the statement follows from the so-called handshaking lemma \citep[Theorem 1.1]{Harris2000}\label{handshaking}, which states that $\sum\limits_{v \in V} deg(v) = 2 \vert E \vert$. Thus, as $G=(V,E)$ is a cubic graph, we have that three times the number of vertices of $G$ equals two times the number of edges in $G$, i.e. $3 \vert V \vert  = 2 \vert E \vert$. Therefore, we have $|V|=\frac{2}{3}|E|$, which (as $|V|$ is an integer) implies that $E=3m$ for some $m \in \mathbb{N}$. Thus, we have $|V|=\frac{2}{3}|E|=\frac{2}{3}\cdot 3m=2m$, which shows that $G$ has an even number of vertices.

Similarly, by applying the handshaking lemma for an unrooted binary phylogenetic network $N=(V,E)$ on $X$ with $|X|=n>1$ , we have that 
$$ 2|E| = \sum\limits_{v \in V} deg(v) = \sum\limits_{v \in X} 1 + \sum\limits_{v \in \mathring{V}} 3 = n + 3 |\mathring{V}|.$$
In particular, $2(|E|-|\mathring{V}|)=2|E|-2|\mathring{V}|=n+|\mathring{V}|=|V|$, and thus the number of vertices in $N$ is even. 
\end{proof}

Another well-known graph theoretical concept we wish to introduce here is a \emph{Hamiltonian path}. A Hamiltonian path is simply a path in a graph that visits each vertex exactly once. If this path is a cycle, the Hamiltonian path is called a \emph{Hamiltonian cycle}. A graph that contains a Hamiltonian cycle is called a \emph{Hamiltonian graph}. As has been noted in \citet{Francis2018} and as we will elaborate in the present manuscript, Hamiltonian paths play an important role concerning the tree-basedness of phylogenetic networks.

\section{Results}
The main aim of this section is to show that all unrooted non-binary phylogenetic networks up to level 3 are tree-based, whereas unrooted non-binary level-4 networks are not necessarily tree-based. However, we first state some preliminary results from the literature for binary networks and show that they also hold for non-binary networks.

\subsection{Basic results} \label{sec_basic}
In this section, we first state some basic results that will be needed throughout this manuscript. Some of them are extensions of the results presented in \citet{Francis2018} from the binary to the non-binary case.

We start by considering the following lemma by \citet{Francis2018}, which states a close relationship for {\em binary } networks between the properties of being tree-based and being proper. In particular, a non-proper phylogenetic network cannot be tree-based. 

\begin{lemma}[\citet{Francis2018}]\label{binarytree-basedproper}
If $N$ is an unrooted binary tree-based network, then every cut edge of $N$ induces a split of $X$, i.e. $N$ is proper.
\end{lemma} 

Note that in the non-binary case, the absence of cut edges need not imply that the network is proper, because it might still contain cut vertices. So we now generalize this lemma to non-binary networks.

\begin{lemma}\label{nonbinarytree-basedproper}
If $N$ is an unrooted tree-based network (binary or not), then $N$ is proper.  
\end{lemma}

\begin{proof} 
Let $N$ be an unrooted tree-based network. Then, it contains a support tree $T$ whose leaf set coincides with $X$. As $T$ is a spanning tree of $N$, its vertex set contains any possible cut vertex $v$, and if $e$ is a cut edge of $N$, it must also occur in the spanning tree (otherwise $T$ could not be connected). Let $T'$ be the base tree corresponding to $T$. Now, as any phylogenetic tree is clearly a proper phylogenetic network, in particular $T'$ is proper. Note that this implies that $T$, too, must have the property that the removal of any cut edge or cut vertex will subdivide $X$, because $T$ is like $T'$, just possibly with additional degree-2 vertices. Thus, removing any cut vertex or cut edge from $T$ results in connected components containing at least one leaf of $T$ (and thus of $N$) each. But as $T$ contains all cut vertices and cut edges of $N$, the same must hold for $N$. In particular, $N$ is proper.
\end{proof}

We next state some general results concerning the decomposition of unrooted networks into simpler parts and the number of vertices in a blob. 

\citet{Francis2018} prove a decomposition theorem for tree-based unrooted binary networks based on the simple networks $B_{N}$ associated with blobs introduced above. This can directly be generalized to non-binary networks and we have the following statement.

\begin{proposition} \label{blob_decomposition}
Suppose $N$ is an unrooted network. Then $N$ is tree-based if and only if $B_{N}$ is tree-based for every blob $B$ in $N$.
\end{proposition}

Our proof of this proposition is very similar to the proof of Proposition 1 in \citet{Francis2018} and Proposition 2.9 in \citet{Hendriksen2018}. However, as we explicitly allow blobs to contain cut vertices, we shortly outline the proof in the Appendix.

Roughly speaking, Proposition \ref{blob_decomposition} states that it is sufficient to analyze all blobs of an unrooted network individually in order to decide whether a network $N$ is tree-based or not. 

As we will show subsequently, it also often makes sense to consider (simple) networks with only two leaves. Therefore, we first recall a useful observation of \citet{Francis2018}, which is, again, also valid for non-binary networks.

\begin{lemma} \label{N-x}
Let $N$ be a network on $X$ with $\vert X \vert \geq 2$. For any $x \in X$ let $N - x$ denote the network obtained from $N$ by deleting $x$ and its incident edge, and suppressing the potentially resulting degree-2 vertex. Then, if $N - x$ is tree-based, so is $N$.
\end{lemma}

The proof of this lemma is very similar to the proof of Lemma 3 in \citet{Francis2018}. However, as we are not only considering binary networks, but also non-binary ones, and as we -- unlike \citet{Francis2018} -- explicitly consider the case of parallel edges, we give the proof again in the Appendix.
In any case it should be noted that $N-x$ might not be a phylogenetic network as it might contain parallel edges. This is why in Definition \ref{def_tree-based} tree-basedness is also defined for (multi)graphs.

\begin{remark}\label{rem_construction}
Note that the converse does not necessarily hold, i.e. if $N$ is tree-based, $N-x$ does not necessarily have to be tree-based, whether $N$ is binary or not. To see this, consider the network depicted in Figure \ref{Fig_nottree-based}. Note that this example is extreme in the following sense: It shows that even if a tree-based network $N$ is binary and contains only one blob and even if $N$ does not contain any cut vertices and no cut edges other than the ones incident to leaves, it might not be possible to remove any leaf and suppress the resulting degree-2 vertex without losing tree-basedness. In particular, there exists no pair of two leaves in this network such that it is still tree-based.

Note that this extreme example is based on the graph shown in Figure \ref{constructionidea}(a), which we found in \citet{Zamfirescu1976}. There, it is proven that any longest path in this graph of 12 vertices has length 10. 

Consider Figure \ref{constructionidea}(b). Note that this graph of 9 vertices, which contains three degree-2 vertices, is Hamiltonian, i.e. it contains a Hamiltonian cycle.\footnote{In fact, as proven in \citet{Zamfirescu1976}, deleting all leaves of this graph corresponds to deleting one vertex from the famous Petersen graph, which is hypohamiltonian.} However, it does \emph{not} contain a Hamiltonian path from either one of these three degree-2 vertices to another one of those: If it did, this path would have length 9 as it would have to cover all 9 vertices, and it could easily be extended by 2 in the graph of Figure \ref{constructionidea}(a) by re-attaching the leaves. This would give a total path of length 11 in the Zamfirescu graph of Figure \ref{constructionidea}(a), which does not exist as the maximal path length there is 10.

Now our network $N$ in Figure \ref{Fig_nottree-based} basically consists of two copies $C_1$ and $C_2$ of the graph of Figure \ref{constructionidea}(b), which are connected by three paths. One of the connecting paths does not have an adjacent leaf. Suppose there was a support tree using only edges of two of these connecting paths. Then this support tree would imply a path through, say, $C_1$, which would imply a Hamiltonian path in the graph of Figure \ref{constructionidea}(b) from one degree-2 vertex to another one -- a contradiction, as we have seen above that such a path cannot exist. Therefore, it is clear that in $N$, edges of all three paths connecting $C_1$ and $C_2$ must be used by any support tree. In fact, as this argument holds for both $C_1$ and $C_2$, in each of the three paths, any support tree must contain an edge incident to $C_1$ \emph{and} an edge incident to $C_2$. As the dashed edge in Figure \ref{Fig_nottree-based} is therefore contained in any support tree of $N$, the induced spanning trees of $C_1$ and $C_2$ are already connected. Thus, using an entire path of the remaining two connecting paths would necessarily induce a cycle, which is not allowed. So there must be at least one edge on both these remaining connecting paths from $C_1$ to $C_2$ which is not covered by any support tree (this is, in fact, the edge between the attachment points of the two leaves of each of these paths). Thus, removing any leaf and suppressing its respective attachment point would necessarily destroy tree-basedness.
\end{remark}

\begin{figure}[htbp]
	\centering
	\includegraphics[scale=0.45]{explongestpath.eps}
	\caption{Unrooted binary tree-based phylogenetic network $N$ on $X=\{x_1,x_2,x_3,x_4\}$. A corresponding support tree is highlighted in bold. $N-x_i$ is not tree-based for $i=1,\ldots,4$, because there is no spanning tree in $N-x_i$ whose leaf set is equal to $X\setminus \{x_i\}$. The dashed edge plays a specific role as explained in Remark \ref{rem_construction}. }
	\label{Fig_nottree-based}
\end{figure}

\begin{figure}[htbp]
	\centering
	\includegraphics[scale=0.35]{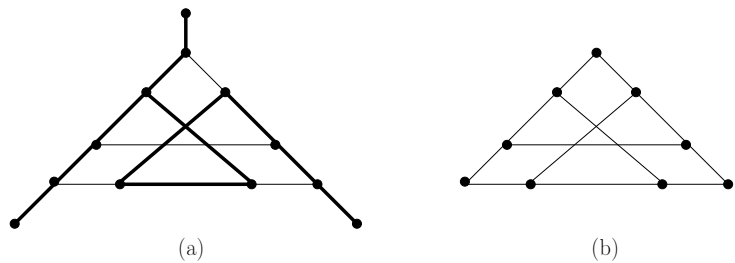}
		\caption{(a) Graph from \citet{Zamfirescu1976} on which our construction of the example in Figure \ref{Fig_nottree-based} is based. There, it is proven that any longest path in this graph of 12 vertices has length 10. However, if we consider this graph as a tree-based phylogenetic network (a support tree is highlighted in bold), removing a leaf vertex and suppressing the resulting degree 2 vertices turns this again into a tree-based phylogenetic network.
		(b) 'Leafless version' of the Zamfirescu graph needed for the construction of Figure \ref{Fig_nottree-based}.}
	\label{constructionidea}
\end{figure}

We now present another observation that will be useful in the following, namely that the number of vertices in a non-trivial blob incident with at most two cut edges is bounded by $2k$ in a level-$k$ network, where $k\geq 2$. 

\begin{lemma} \label{blob_nodes}
Let $N$ be an unrooted level-$k$ network (not necessarily binary) with $k\geq 2$, let $B$ be a non-trivial blob, and assume that there are at most two cut edges in $N$ incident to $B$. 
Let $n = \vert V(B) \vert$ denote the number of vertices in $B$. Then, 
$$ 3 < n \leq 2k.$$ 
\end{lemma}

\begin{proof}
We first show that $n > 3$. Let $m = \vert E(B) \vert$ denote the number of edges in $B$.
As we assumed that there are at most two cut edges in $N$ incident to $B$, $B$ can have at most two degree-2 vertices. All other vertices are of degree at least $3$ (they can be higher, as we do not assume that the network is binary).  Thus, using the handshaking lemma, we have
	\begin{equation*}
	m = \frac{1}{2} \sum\limits_{v \in V(B)} deg(v) 
	\geq \frac{1}{2} \big(2 \cdot 2 + (n-2) \cdot 3 \big) 
	= \frac{1}{2} \big(3n - 2) 
	= \frac{3}{2} n -1.
	\end{equation*}
On the other hand, as we are considering simple graphs, $B$ can have at most ${n \choose 2}=\frac{n(n-1)}{2}$ edges. Thus,
$$ \frac{3}{2} n -1 \leq m \leq \frac{n(n-1)}{2},$$  
which leads to $$n^2-4n+2 \geq 0.$$ This inequality is fulfilled if $n\geq2+\sqrt{2}$ or if $n\leq 2-\sqrt{2}$. As $n$ is a positive integer, we conclude $n > 3$. \\

Next, we need to show that $n \leq  2k$ for $k \geq 2$. Recall that an unrooted tree on $n$ vertices has $n-1$ edges. 
As $N$ is a level-$k$ network, at most $k$ edges have to be removed from $B$ to obtain a tree.
Thus, $m - (n-1) \leq k$. In other words, $m \leq k+n-1$.
Using the lower bound for $m$ from above, we derive
\begin{align*}
\frac{3}{2} n -1 \leq m \leq k+n-1, 
\end{align*}
and thus in particular $n \leq 2k$. This completes the proof. 
\end{proof}

\subsection{Unrooted non-binary tree-based networks} \label{sec_non-binary}
We are now in the position to state the main result of this paper, namely the fact that all proper unrooted non-binary phylogenetic networks up to level 3\footnote{Note that by definition, all networks of level less than 3 are contained in the class of level-3 networks.} are tree-based, whereas unrooted non-binary level-4 networks are not necessarily tree-based.

\begin{theorem} \label{level_non-binary}
All proper unrooted non-binary unrooted level-3 are tree-based. Moreover, unrooted non-binary networks of level greater than 3 need not be tree-based.
\end{theorem}

In order to prove Theorem \ref{level_non-binary}, we need the following lemma, which gives a lower bound on the number of vertices needed to destroy tree-basedness.

\begin{lemma} \label{minima_non-binary}
Any minimal (in the number of vertices) proper unrooted non-binary non-tree-based network has 6 interior vertices and 2 leaves, i.e. 8 vertices in total.
\end{lemma}

\begin{proof}
It suffices to show that all proper unrooted non-binary networks on less than 8 vertices are tree-based. In this regard, we performed an exhaustive search using Mathematica \citep{Mathematica}. In the following we explain the details of this exhaustive search.

First of all, we obtained a list containing all simple graphs with up to 7 vertices from the \enquote{House of Graphs} database (cf. \citet{Brinkmann2013}).

We then filtered these graphs for networks and thus excluded unconnected graphs, graphs without leaves (i.e. vertices of degree 1)
as well as graphs with vertices of degree 2. 

We then further excluded all non-proper networks by checking whether the removal of any cut edge or any cut vertex present in the network resulted in connected components containing at least one leaf each. 

This left us with 28 proper unrooted non-binary networks, which can all easily be shown to be tree-based. We depicted all of them in the Appendix (see catalog of all proper tree-based networks with up to 7 vertices in Figures \ref{Catalog1} and \ref{Catalog2}).

This shows that all unrooted non-binary networks with up to 7 vertices are tree-based. However, for 8 vertices there exist unrooted non-binary networks which are not tree-based. One such example is depicted in Figure \ref{Fig_4l-non-binary-nontreebas}. This network has 6 interior vertices and 2 leaves and it can be seen not to be tree-based as follows. If it was tree-based, then there would be a path from $x$ to $y$ visiting each vertex exactly once. Trivially, this path must start with edge $\{x,a\}$ and end with edge $\{f,y\}$. Now, the first three vertices of the path must either be $(x,a,b), (x,a,d)$ or $(x,a,e)$. In all cases, the 4th vertex must be $c$, because the only other vertex (except from $a$) that can be reached from $b,d$ or $e$ is $f$; however, $f$ has to be visited second last. Thus, so far we either have a path starting $(x,a,b,c), (x,a,d,c)$ or $(x,a,e,c)$. Consider the first case, i.e. $(x,a,b,c)$. From $c$ we can either go to $d$ or to $e$. If we go to $d$, i.e. if we have $(x,a,b,c,d)$, the only \enquote{free} vertex reachable from $d$ is $f$. This leads to a contradiction as vertex $e$ has not been visited.
If we instead go to $e$, i.e. if we have $(x,a,b,c,e)$, the only \enquote{free} vertex reachable from $e$ again $f$, which causes a contradiction as $d$ has not been visited. Similar contradictions follow for all other cases.
This completes the proof.
\end{proof}

\begin{figure}[hbtp]
	\centering
	\includegraphics[scale=0.2]{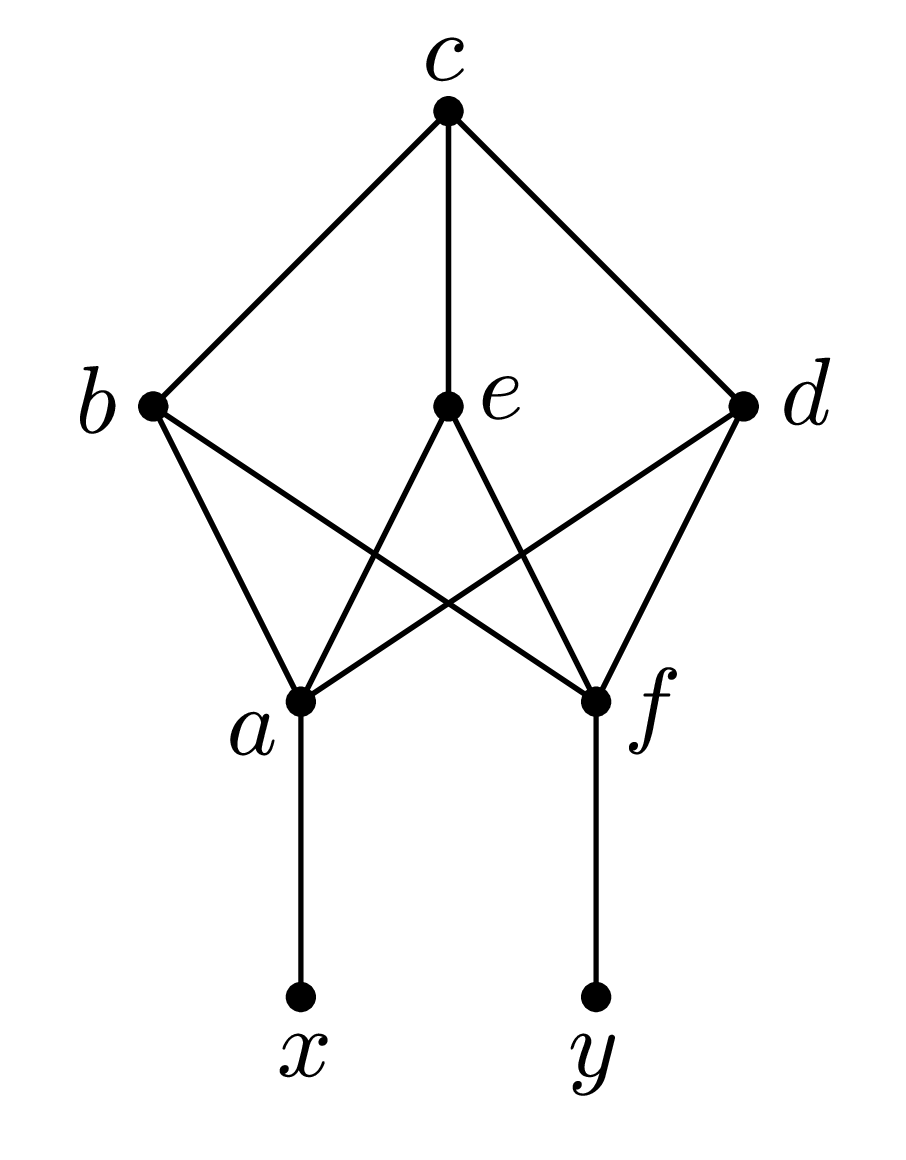}
	\caption{Unrooted non-binary level-4 network that is not tree-based (this network is adapted from \citet{Jetten2018}, where it is used in a different context).}
	\label{Fig_4l-non-binary-nontreebas}
\end{figure}

We are now in the position to prove Theorem \ref{level_non-binary}.

\begin{proof}[Proof (Theorem \ref{level_non-binary})]
We have to show that all proper unrooted non-binary phylogenetic networks of levels $k=0, \ldots, 3$ are tree-based. Note that if $X$ contains only one leaf, either the underlying network consists of a single vertex (and is thus trivially tree-based) or it is not proper (Remark \ref{all_leaves_one_vertex}). This is why we now consider only networks with $|X|\geq 2$. 

Recall that we can decompose $N$ into a collection of simple networks $B_N$ associated with the non-trivial blobs in $N$, each having at least two leaves, and if each of these simple networks is tree-based, then so is $N$ by Proposition \ref{blob_decomposition}.  (Note that trivial blobs are generally tree-based, which is why we only consider non-trivial blobs and their associated networks $B_N$).
Moreover, if we remove all but 2 leaves from each of these simple networks $B_N$ and obtain a tree-based network, then $B_N$ must have been tree-based due to Lemma \ref{N-x}.

Now, for $k=0$, the statement holds, as a level-0 network is a tree and thus tree-based. For $k=1$, we know that at most one edge has to be removed from each non-trivial blob of $N$ to obtain a tree. Removing at most one edge from each such non-trivial blob of $N$, however, cannot induce any new leaves (if an interior vertex of some blob became a leaf after removing one edge from this blob, this would imply that this vertex is a degree-2 vertex in $N$, which is not allowed. Thus, the tree that we obtain from removing at most one edge from each non-trivial blob in a level-1 network can directly be considered a support tree for $N$. This implies that level-1 networks are always tree-based.

Let us now consider $k \geq 2$. By Lemma \ref{minima_non-binary}, we know that any simple network $B_{N}$ that leads to a proper non-binary non-tree-based network has to have at least 6 interior vertices and 2 leaves, i.e. at least 8 vertices in total. Moreover, by Lemma \ref{blob_nodes} we know that for $k \geq 2$, the number of vertices in a non-trivial blob associated with a network $B_{N}$ is bounded from above by $2k$, i.e. in case of level 2, the maximum number of vertices in $B_N$ is 4. This implies that there cannot exist a level-2 network that is not tree-based. 

In order to prove that there also does not exist a non-tree-based level-3 network, we exhaustively generated all unrooted proper phylogenetic networks on 8 vertices, analyzed them for tree-basedness and computed their level.
This exhaustive search was conducted in the following way: We used the \enquote{House of Graphs} database (cf. \citet{Brinkmann2013}) to obtain a list of all graphs with 8 vertices (12346 in total). These were then filtered for unrooted proper phylogenetic networks in the same way as described in the proof of Lemma \ref{minima_non-binary}, resulting in a total of 197 unrooted proper phylogenetic networks. These were then analyzed for tree-basedness and it turned out that there are only 8 proper phylogenetic networks on 8 vertices that are \emph{not} tree-based (see Figure \ref{Fig_Catalog8}). However, none of them is a level-3 network, i.e. we can conclude that all proper non-binary level-3 networks are indeed tree-based.

To prove the last statement of the theorem, consider the unrooted non-binary level-4 network depicted in Figure \ref{Fig_4l-non-binary-nontreebas}. This network is not tree-based as we have already seen in the proof of Lemma \ref{minima_non-binary}. This completes the proof.
\end{proof}

\begin{remark}
Note that the last part of the proof above answers a question posed in \citet{Hendriksen2018}, asking whether there exist networks of level less than 5 that are not loosely tree-based. As the definition of loosely tree-based in \citet{Hendriksen2018} precisely corresponds to our definition of tree-based, the level-4 network depicted in Figure \ref{Fig_4l-non-binary-nontreebas} provides such an example. 
\end{remark}

\subsection{Unrooted binary tree-based networks} \label{sec_binary}
We now re-visit  tree-basedness of unrooted \emph{binary} phylogenetic networks, in order to re-establish a known result and in order to extend it by an additional proposition.

Recall that \citet{Francis2018} recently established the following result, which is the binary analog of our Theorem \ref{level_non-binary}.

\begin{theorem}[adapted from \citet{Francis2018}] \label{level4_theorem}
All proper unrooted binary level-4 network are tree-based. Moreover, networks of level greater than 4 need not be tree-based.
\end{theorem}

Note that the example given in \citet{Francis2018} for the second part of this theorem was unfortunately erroneous. Indeed, the level-5 network on $X=\{x,y\}$ depicted in Figure \ref{Fig_wrongExample} was used in \citet{Francis2018} to show that level-5 networks need not be tree-based. However, this network in fact \emph{is} tree-based as there exists a spanning tree in $N$ whose leaf set is equal to $X$. While the authors have recently independently pointed out this error and published an erratum \citep{Francis2018erratum}, in the process of writing this manuscript we have further analyzed unrooted binary non-tree-based networks. Using a similar approach as in our proof of Theorem \ref{level_non-binary}, we provide an alternative proof of Theorem \ref{level4_theorem}, the details of which are given in the Appendix. More importantly, though, this  analysis enabled us to show that binary non-tree-based level-5 networks are `rare' in the sense stated by the following proposition.

\begin{figure}[htbp]
	\centering
	\includegraphics[scale=0.3]{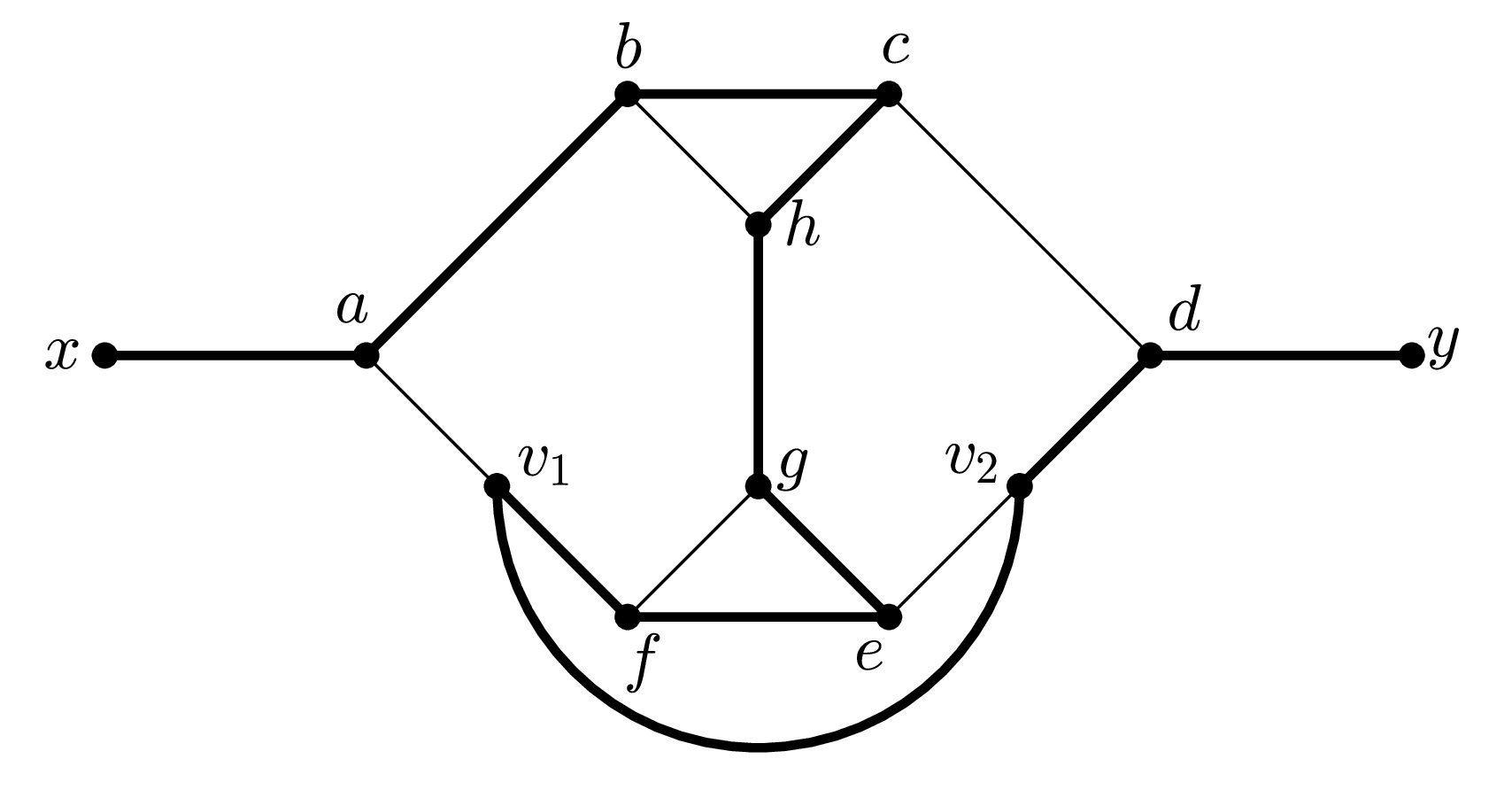}
	\caption{Level-5 network on $\{x,y\}$ claimed to not be tree-based in \citet{Francis2018}. However, this network is tree-based, because there is a spanning tree in $N$ whose leaf set is equal to $\{x,y\}$ (depicted in bold). This spanning tree is a path between $x$ and $y$ consisting of the following edges: $\{x,a\}, \, \{a,b\}, \, \{b,c\}, \, \{c,h\}, \, \{h,g\}, \, \{g,e\}, \, \{e,f\}, \, \{f,v_1\}, \, \{v_1,v_2\}, \, \{v_2,d\}$ and $\{d,y\}$.}
	\label{Fig_wrongExample}
\end{figure}

\begin{proposition}\label{prop_2networks_only}
There are precisely two minimal (in the number of vertices) proper unrooted binary phylogenetic level-5 networks that are not tree-based. Both of them contain exactly 12 vertices. In particular, all other proper unrooted binary phylogenetic level-5 networks that are not tree-based contain more than 12 vertices.
\end{proposition}

The proof of this proposition is given in the Appendix. The two minimal networks, however, are depicted in Figure \ref{Fig_correctExample}. 

\begin{remark} Note that both networks in \ref{Fig_correctExample} are symmetric in the sense that swapping their leaves $x$ and $y$ leads to an isomorphic network. So in fact, there are really only two minimal  \emph{phylogenetic} networks fulfilling the required properties from Proposition \ref{prop_2networks_only}, even though our proof technique described in the Appendix is based on unlabeled graphs.  \end{remark}

\section{Discussion} \label{sec_discussion}  
The main aim of this manuscript was to generalize some of the results presented in \citet{Francis2018} for binary unrooted tree-based networks to non-binary ones. While some results are straightforward to generalize, others need to be adjusted. Differences between non-binary and binary networks include, for instance, that unlike in the binary case, level-4 networks need not necessarily be tree-based in the non-binary case. This provides the answer to Question 5.3 in \citet{Hendriksen2018}, asking whether there are non-binary networks of level less than 5 that are not tree-based.

While the main focus of this manuscript was on non-binary networks we also re-visited binary ones. In particular, we gave two examples showing that binary level-5 networks are not always tree-based, as the example given in \citet{Francis2018} was unfortunately erroneous. Even though one correct example has in the meantime independently been published by the authors in an erratum (cf. \citet{Francis2018erratum}), our study additionally shows that binary level-5 networks that are not tree-based are `rare' in the sense that there are only two minimal ones. Furthermore, our study provides an alternative proof for the fact that all proper unrooted binary networks up to level 4 are tree-based (Theorem 1 in \citet{Francis2018}).

Tree-based phylogenetic networks are an interesting topic for future research. Note that they are not only of high biological relevance -- they also link mathematical phylogenetics to classic graph theory, like e.g. Hamiltonian paths and cycles, cubic graphs, and classic results like the ones in \citet{Zamfirescu1976}. But not only does the study of phylogenetic networks benefit from connections to classic graph theory -- in fact, the opposite is also true: In graph theory and various applications, finding spanning trees with few leaves has been a topic of high interest (cf. \citet{tsugaki2007,Rivera_Campo_2012,SALAMON2008164}). Note that every connected graph has a spanning tree, and in all spanning trees, any vertices of the graph of degree 1 must be contained as leaves. So for phylogenetic networks, it is obvious that all leaves of the network are necessarily also leaves of any spanning tree. Tree-basedness, however, implies that there are no more leaves than this minimum requirement suggests. In this regard, support trees of tree-based phylogenetic networks are just spanning trees with the smallest possible number of leaves. Many techniques, like e.g. Hamiltonicity, which have already been used to classify spanning trees with few leaves (cf. \citet{SALAMON2008164}) are therefore also useful in the study of tree-basedness. We are sure that this connection to graph theory will prove fruitful for future studies, too.

\begin{figure}[htbp]
	\centering
	\includegraphics[scale=0.25]{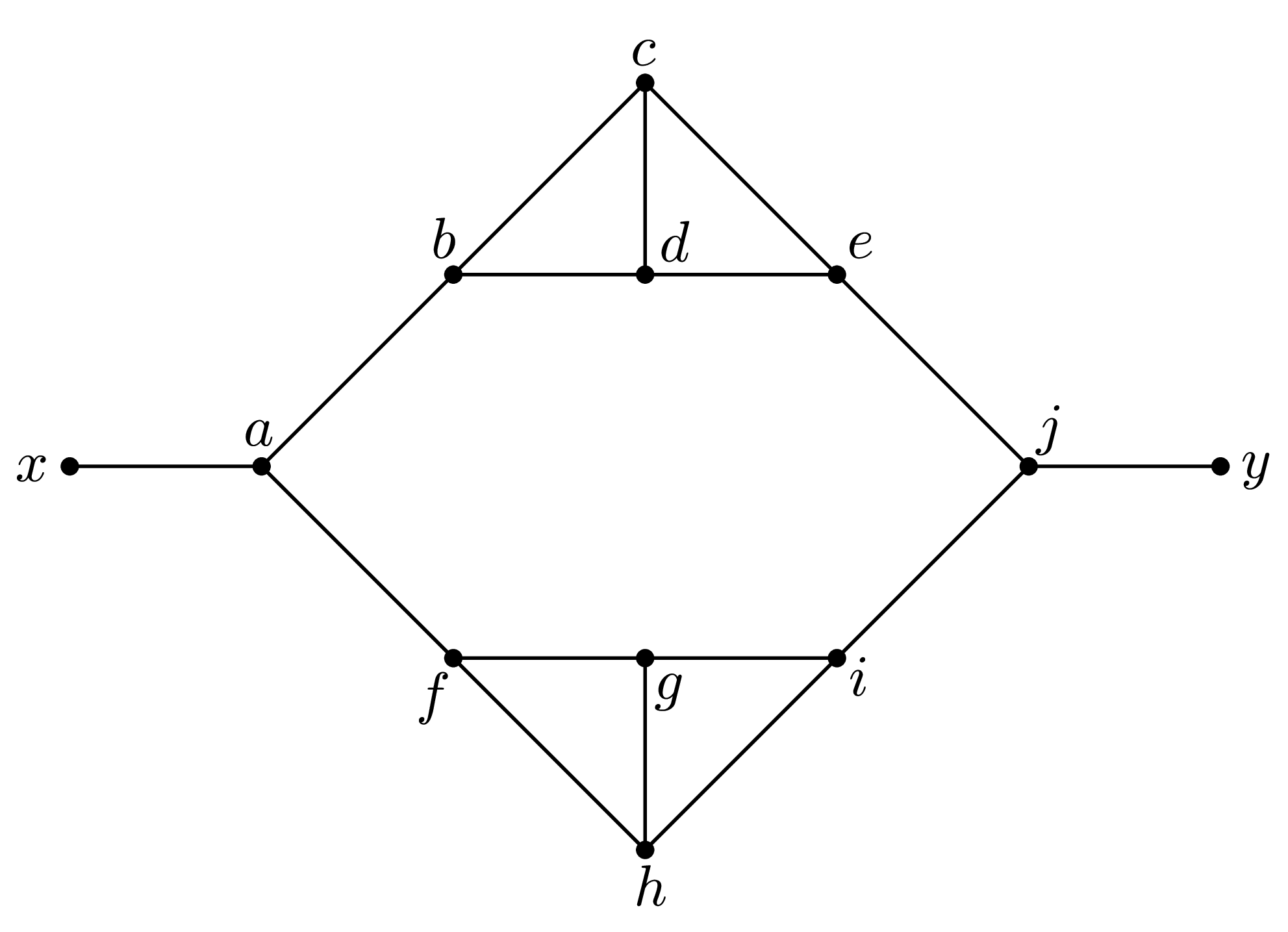} 
	\includegraphics[scale=0.25]{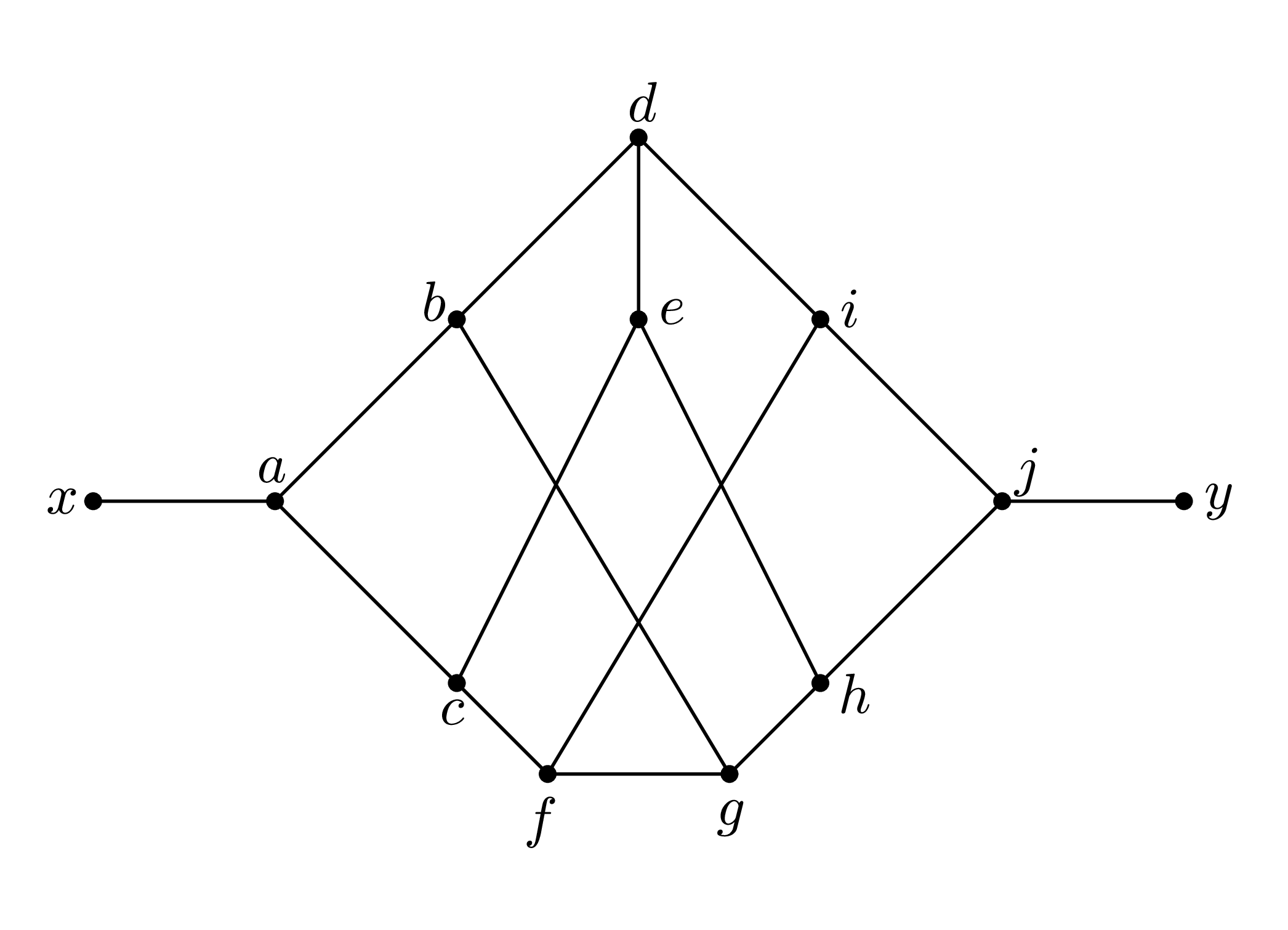}
	\caption{The only two proper level-5 networks on $X=\{x,y\}$ with 12 vertices that are not tree-based.}
	\label{Fig_correctExample}
\end{figure}

\section*{Acknowledgement} 
We thank two anonymous reviewers of an earlier version of this manuscript for their valuable suggestions and feedback. Moreover, the first author thanks the joint research project {\bf \emph{DIG-IT!}} supported by the European Social Fund (ESF), reference: ESF/14-BM-A55-0017/19, and the Ministry of Education, Science and Culture of Mecklenburg-Vorpommern, Germany. Moreover, the fourth author was supported by the National Science Foundation for Young Scientists in China (Grant No. 11901227).
Additionally, the third and fifth author thank the state Mecklenburg-Western Pomerania for the Landesgraduierten-Studentship. Moreover, the second author thanks the University of Greifswald for the Bogislaw-Studentship and the fifth author thanks the German Academic Scholarship Foundation for a studentship.

\bibliographystyle{abbrvnat}
\bibliography{References2.bib}

\section{Appendix} \label{sec_Appendix}
\setcounter{figure}{0}
\renewcommand{\thefigure}{\Roman{figure}}

\subsection{Supplementary lemmas}
\begin{lemma}\label{binarycutvertex}
Let $N$ be an unrooted binary phylogenetic network. Then every cut vertex is incident to a cut edge.
\end{lemma}
\begin{proof}
Let $N$ be an unrooted binary phylogenetic network with cut vertex $v$.
Then, $v$ has degree three, because leaves cannot be cut vertices. We call the three edges incident to $v$ $e_1, e_2$ and $e_3$ as depicted in Figure \ref{threecomponents}.
\begin{figure}[H]
\centering
\includegraphics[scale=0.4]{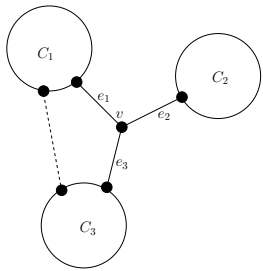}
\caption{Unrooted binary phylogenetic network $N$ with cut vertex $v$ and edges $e_1,e_2$ and $e_3$ leading to components $C_1,C_2$ and $C_3$.}
\label{threecomponents}
\end{figure}
It is possible that $e_1, e_2$ and $e_3$ lead to three different components $C_1, C_2$ and $C_3$ (in this case the dashed line in Figure \ref{threecomponents} is excluded). Then, $e_1, e_2$ and $e_3$ are all cut edges and thus $v$ is incident to three different cut edges. 
\\
Otherwise, two edges lead to the same component. Without loss of generality $e_1$ and $e_3$ lead to the same component (dashed line in Figure \ref{threecomponents}). Then $e_2$ is a cut edge. Therefore, $v$ is incident to a cut edge.
\\
Note that not all edges can lead to the same component, because if that was the case $v$ would not be a cut vertex. This completes the proof.
\end{proof}

\setcounter{proposition}{1}
\begin{proposition}
Suppose $N$ is an unrooted network. Then $N$ is tree-based if and only if $B_{N}$ is tree-based for every blob $B$ in $N$.
\end{proposition}

\begin{proof} Suppose $N$ is an unrooted tree-based network on $X$. As $N$ is tree-based there exists a support tree $T$ for $N$, i.e. a spanning tree with leaf set $X$. As $T$ is a spanning tree, $T$ in particular contains all cut vertices of $N$. Moreover, it has to contain all cut edges of $N$, because otherwise $T$ would not be connected. Thus, any support tree for $N$ induces a spanning tree of $B_{N}$ and we can conclude that $B_{N}$ is tree-based. Conversely, suppose that $B_{N}$ is tree-based for every blob $B$ in $N$. Then by taking a support tree for $B_{N}$ for each blob, we can construct a support tree $T$ for $N$ by connecting the individual support trees corresponding to the blobs of $N$ via the cut-edges that connected the blobs of $N$. Thus, $N$ is tree-based.
\end{proof}

\setcounter{lemma}{2}
\begin{lemma}
Let $N$ be a network on $X$ with $\vert X \vert \geq 2$. For any $x \in X$ let $N - x$ denote the network obtained from $N$ by deleting $x$ and its incident edge, and suppressing the potentially resulting degree-2 vertex. Then, if $N - x$ is tree-based, so is $N$.
\end{lemma}

\begin{proof} 
Let $N$ be a network on $X$ with $\vert X \vert \geq 2$. Let $N-x$ be obtained from $N$ by deleting leaf $x$ and its incident edge. If this results in a degree-2 vertex $v$ (which for example is the case if $N$ is a binary network), $v$ is suppressed. Note that this might imply that $N-x$ contains parallel edges (cf. Figure \ref{fig_suppressing}) and thus, is not a phylogenetic network anymore. However, the leaf deletion cannot result in an  unconnected (multi)graph. In particular, $N-x$ is a connected (multi)graph. We will now show that if $N-x$ is tree-based, so is $N$. At this point, it is important to notice that in this manuscript tree-basedness is not only defined for phylogenetic networks, but more generally for (multi)graphs (cf. Definition \ref{def_tree-based}).
Let $T$ be a support tree for $N-x$. We now distinguish between three cases:
	\begin{enumerate}
	\item $deg(v) = 1$ in $N$: \\
	If $deg(v) = 1$ in $N$, $N$ must consist of a single edge, namely $\{v,x\}$. A single edge is trivially tree-based, so there is nothing to show. 
	\item $deg(v)$ in $N$ is strictly greater than 3: \\
	If $deg(v)$ in $N$ is strictly greater than 3, it is strictly greater than 2 in $N-x$. This implies that $v$ is not suppressed in $N-x$. Then, we can obtain a support tree for $N$ from $T$ by adding the edge $\{x,v\}$ to $T$.
	\item $deg(v)=3$ in $N$, i.e. $deg(v)=2$ in $N-x$: \\
	Let $v_1,v_2 \neq x$ denote the other two vertices adjacent to $v$ in $N$. Let $\{v_1, v_2\}$ denote the edge that results from suppressing $v$ in $N-x$. Note that this might be a parallel edge, if there already is an edge $\{v_1,v_2\}$ in $N-x$. Now, there are two cases:
		\begin{itemize}
		\item If $\{v_1,v_2\}$ is an edge in $T$ (if there are multiple edges between $v_1$ and $v_2$, $T$ will only contain one of them), then we can obtain a support tree for $N$ by subdividing this edge (i.e. re-introducing the attachment point $v$) and adding the edge $\{x,v\}$ to $T$.
		\item If $\{v_1,v_2\}$ is not an edge in $T$, we note the following. As $T$ is a support tree for $N-x$, it must contain both $v_1$ and $v_2$. Thus, we can obtain a support tree for $N$ by re-introducing vertex $v$ and the edges $\{v_1,v\}$ and $\{v,x\}$ (or $\{v_2,v\}$ and $\{v,x\}$) to $T$.
		\end{itemize}
	\end{enumerate}
This completes the proof.
\end{proof}

\subsection{Analysis of unrooted binary networks and alternative proof of Theorem 1 from \citet{Francis2018}} \label{sec_alternativeproof}

In the following we provide an alternative proof of Theorem 1 from \citet{Francis2018}. This alternative approach then also allows us to conclude that there are precisely 2 minimal proper unrooted binary level-5 networks that are not tree-based.

\setcounter{theorem}{1} 
\begin{theorem}[adapted from \citet{Francis2018}]
All proper unrooted binary level-4 networks are tree-based. Moreover, networks of level greater than 4 need not be tree-based.
\end{theorem}

However, we first need to introduce further definitions and notations. 
We begin by describing the so-called \emph{leaf connecting procedure}, which was recently introduced in \citet{paper1}.

Let $N$ be a phylogenetic network that is not a tree, 
with taxon set $X$ with $|X| \geq 2$, i.e. $N$ contains at least two leaves. The aim of the leaf connecting procedure is to turn $N$ into a graph without leaves, i.e. without degree-1 vertices. This is achieved in the following way (cf. \citet{paper1}):
	\begin{itemize}
	\item Pre-processing: 
	As long as there exists an interior vertex $u$ of $N$ such that there is more than one leaf attached to $u$, delete all of them but one. Additionally, suppress potentially resulting degree-2 vertices. In the following, we denote the resulting reduced taxon set of $N$ by $X^r$.
	\end{itemize}
	Note that this pre-processing step may have to be repeated several times, but does not influence whether a network is tree-based or not (cf. \citet{paper1}). Moreover, it is worth mentioning that if $N$ was a tree, the pre-processing step would necessarily result in a single edge, and the following leaf connecting procedure could thus not take place. This is why we consider only non-tree networks here.
	
	\begin{itemize}
	\item Leaf connecting: 
	\begin{itemize}
	\item Select two leaves $x_1$ and $x_2$ (if they exist) and denote their respective attachment points by $u_1$ and $u_2$, respectively. Now, delete $x_1$ and $x_2$ as well as their incident edges (i.e. the edges $\{x_1,u_1\}$ and $\{x_2,u_2\}$) and connect their attachment points by introducing a new edge $e:=\{u_1,u_2\}$. If this edge is a parallel edge, i.e. if there is another edge $\widetilde{e}$ connecting $u_1$ and $u_2$, add two more vertices $a$ and $b$ and replace $e$ by two new edges, namely $e_1:=\{u_1,a\}$ and $e_2:=\{a,u_2\}$. Similarly, replace $\widetilde{e}$ by two new edges, namely $\widetilde{e}_1:=\{u_1,b\}$ and $\widetilde{e}_2:=\{b,u_2\}$. Last, add a new edge $\{a,b\}$. \\
	Repeat this procedure until no pair of leaves is left.
	\item If there is one more leaf $x$ left in the end, remove $x$ and, if its attachment point $u$ then has degree 2, suppress $u$. If this results in two parallel edges, say $e=\{y,z\}$ and $\widetilde{e} = \{y,z\}$, re-introduce $u$ on edge $e$ and add a new vertex $a$ to the graph, delete $\widetilde{e}$ and introduce two new edges $\widetilde{e}_1:=\{y,a\}$ and $\widetilde{e}_2:=\{a,z\}$. Last, add an edge $\{u,a\}$.
	\end{itemize}		
\end{itemize}
Note that the order in which the leaves are connected may have an impact on the resulting graph. In general, if $|X| > 2$, there might be more than one graph that can be constructed from $N$ by the leaf connecting procedure. We denote the set of all these graphs by $\mathcal{LCON}(N)$.
An illustration of this concept is given in Figure \ref{Fig_LeafConnecting_new}.

\begin{figure}
	\centering
	\includegraphics[scale=0.15]{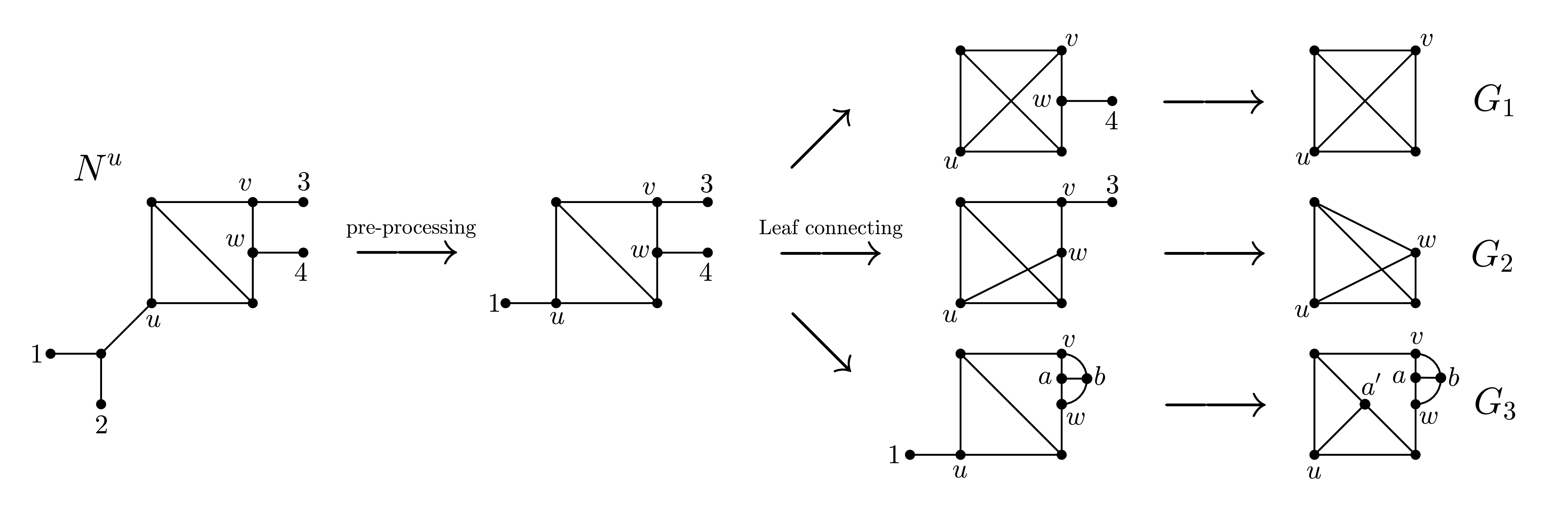}
	\caption{Network $N$ on taxon set $X=\{1,2,3,4\}$ and the graphs resulting from the leaf connecting procedure. In the pre-processing step leaf $2$ is removed and the resulting degree-2 vertex is suppressed. Then, first a pair of leaves is chosen and removed from the network, their attachment points are connected and, if necessary, new vertices and edges are introduced. Lastly, the remaining single leaf is removed, and again, if necessary new vertices and edges are introduced. This results in three graphs: $G_1$, $G_2$ and $G_3$. Note, however, that $G_1$ and $G_2$ are isomorphic. Thus, $\mathcal{LCON}(N)$ consists only of $G_1$ and  $G_3$. Moreover, note that even though new vertices were introduced to obtain $G_3$, the total number of vertices in the graph did not increase, i.e. $G_3$ contains only as many vertices as $N$ after the pre-processing step.}
	\label{Fig_LeafConnecting_new}
\end{figure}

Before we can proceed with the proof of the first part of Theorem \ref{level4_theorem}, we briefly outline our proof strategy: As in the proof of Theorem \ref{level_non-binary}, we use the fact that it is sufficient to consider the non-trivial blobs of the unrooted binary network $N$ and show that all such networks with only two leaves have a close relationship with cubic graphs via the $\mathcal{LCON}(N)$ construction, which in this case, where $N$ has only two leaves, contains a unique cubic graph which we will call $G(N)$. Moreover, we show if such a network is tree-based, $G(N)$ needs to have a Hamiltonian cycle and thus $N-X$ must contain a Hamiltonian path between the two attachment points of its leaves. 

Note that if such a network with two leaves is tree-based, we can simply attach more leaves by Lemma \ref{N-x} without losing the tree-basedness. So any network that is not tree-based can in particular {\em not} contain a subnetwork with two leaves which are connected by a Hamiltonian path. We can thus investigate Hamiltonian paths in cubic graphs a bit more in-depth and use a simple counting argument based on Lemma \ref{blob_nodes} to show that the number of vertices necessary to avoid a Hamiltonian path induces a level of $k\geq 5$. 

We begin with establishing the required relationship between unrooted binary phylogenetic networks with two leaves and cubic graphs:

\begin{observation} \label{cubic_tree-based} Let $N$ be a proper unrooted binary phylogenetic network on leaf set $X$ with $|X|=2$ and with $\mathring{E}\neq \emptyset$. Without loss of generality, let $X=\{x,y\}$ and denote the vertices adjacent to $x$ and $y$ by $u$ and $v$, respectively. Then,  $\mathcal{LCON}(N)$ contains precisely one graph $G(N)$, and this graph is cubic. Moreover, by construction the number of vertices of $G(N)$ is bounded by the number of vertices of $N$, i.e. we have $|V(G(N)) |\leq |V(N) |$, as in each step, a leaf and its attachment point get deleted, and at most two new vertices get introduced (if otherwise we would have a parallel edge). 
\end{observation}

Note that the construction of $G(N)$ does not require the suppression or deletion of any vertices other than $x$ and $y$ (as $N$ is proper, $x$ and $y$ cannot be attached to the same interior vertex, so there cannot be need for the pre-processing step; furthermore, as $N$ contains precisely two leaves ($x$ and $y$), no single remaining leaf needs to be removed in the end) and so, as we require $\mathring{E}\neq \emptyset$, $N$ cannot simply be a tree consisting of two vertices connected by a single edge. This implies that the resulting graph $G(N)$ is {\em always} cubic. Moreover, we state the following crucial proposition.

\setcounter{proposition}{3}
\begin{proposition} \label{cubic_tree-based2}
Let $N$ and $G(N)$ be as described in Observation \ref{cubic_tree-based}. Then, $N$ is tree-based if and only if $G(N)$ contains a Hamiltonian cycle using at least one edge of $G(N)$ that is not contained in $N$. 
\end{proposition}

\begin{proof}
Let $N$ be an unrooted binary phylogenetic network on leaf set $X=\{x,y\}$ and with $| \mathring{E} | \neq \emptyset$. Let $u$ and $v$ denote the vertices adjacent to $x$ and $y$, respectively. Consider the graph $G(N)$ obtained from the leaf connecting procedure. (Note that $G(N)$ might contain two new vertices, $a$ and $b$, if in the construction of $G(N)$ parallel edges occurred.)
Now, assume that $G(N)$ contains a Hamiltonian cycle using at least one edge of $G(N)$ that is not contained in $N$.
We now distinguish between two cases:
\begin{itemize}
\item $G(N)$ does not contain new vertices $a$ and $b$: \\
As $G(N)$ contains a Hamiltonian cycle using the edge $\{u,v\}$ that is not contained in $N$, this implies that there is a Hamiltonian path from $u$ to $v$ in $G(N)$. We extend this path to a support tree of $N$ by adding $x$ and $y$ as well as the edges $\{x,u\}$ and $\{y,v\}$. 
\item $G(N)$ contains new vertices $a$ and $b$: \\
As $G(N)$ contains a Hamiltonian cycle using either the edges $\{u,a\}, \, \{a,b\}, \, \{b,v\}$ or $\{u,b\}, \, \{a,b\}, \, \{a,v\}$, deleting the edge $\{a,b\}$ and suppressing $a$ and $b$ as well as one copy of the parallel edges $e=\tilde{e}=\{u,v\}$ results in a Hamiltonian path from $u$ to $v$ in this modified network. As before, we can now extend this path to a support tree of $N$ by adding $x$ and $y$ as well as the edges $\{x,u\}$ and $\{y,v\}$. This completes the first direction of the proof. 
\end{itemize}
On the other hand, if $N$ is tree-based, this implies that there is a spanning tree whose leaf set is precisely $X=\{x,y\}$. 
Again, we distinguish between two cases:
\begin{itemize}
\item If the construction of $G(N)$ does not require adding $a$ and $b$, it immediately follows that the support tree of $N$ leads to a Hamiltonian cycle in $G(N)$, as we can go from $u$ to $v$ both via the support tree, which covers all vertices of $G(N)$, or via the new edge $\{u,v\}$. Thus, we have a Hamiltonian cycle which uses a new edge.
\item If the procedure requires the introduction of $a$ and $b$, we can obtain a Hamiltonian cycle in $G(N)$ by extending the support tree of $N$ by the edges $\{u,a\}, \, \{a,b\}$ and $\{b,v\}$ and removing the edges $\{x,u\}$ and $\{v,y\}$ from the support tree. In particular, this cycle uses three new edges.
\end{itemize}
This completes the proof.
\end{proof}

We now state the following Lemma stating that any minimal proper unrooted binary non-tree-based network contains 12 vertices. This lemma is the binary analog of Lemma \ref{minima_non-binary} in the main part of the manuscript. It shows that in the binary case, more vertices are required (namely at least 12) in order to destroy tree-basedness than in the non-binary case (where we need at least 8).

\setcounter{lemma}{6}
\begin{lemma}\label{12vertices}
Any minimal proper unrooted binary non-tree-based network has 12 vertices (10 interior vertices and 2 leaves).
\end{lemma}

\begin{remark} Note that by the proofs of Theorem \ref{level4_theorem} and Proposition \ref{prop_2networks_only}, the bound of 12 vertices is tight, i.e. there are networks that achieve it. 
\end{remark}

In order to prove the lemma, we require the following statement:
\begin{lemma} \label{cubic_hamiltonian}
Let $G=(V,E)$ be a cubic graph with $\vert V \vert \leq 8$. Then, $G$ contains a Hamiltonian path from $u$ to $v$ for all edges $e=\{u,v\} \in E$. In other words, $G$ is Hamiltonian and for every edge $e \in E$, there is a Hamiltonian cycle of $G$ which contains $e$. 
 \end{lemma}

\begin{proof} As the number of vertices in a cubic graph is even by Proposition \ref{prop_cubic_even} and as the smallest cubic graph contains four vertices, we only need to consider the one cubic graph with four vertices, the two cubic graphs with six vertices and the five cubic graphs with eight vertices (all these graphs are depicted in the Appendix in Figures \ref{figcubic1}, \ref{figcubic2} and \ref{figcubic3}). The fact that they are all Hamiltonian can be found in the literature (cf. \citet{Bussemaker1976}), but can also easily be verified by considering all mentioned 8 graphs. We also verified the fact that each edge is contained in at least one Hamiltonian cycle exhaustively.
This completes the proof.
\end{proof}

We can now prove Lemma \ref{12vertices}.

\begin{proof}[Proof of Lemma~\ref{12vertices}] First recall that the number of vertices in an unrooted binary phylogenetic network is always even unless if $N$ simply consists of only one vertex (cf. Proposition \ref{prop_cubic_even}). 

Now assume that $N$ is a proper unrooted binary non-tree-based network. As $N$ is non-tree-based, in particular $N$ does not only consist of one vertex or of two vertices connected by an edge. So as the number of vertices has to be even, the total number of vertices has to be at least four.

Moreover, consider $|X|$. If $|X|=1$, then $N$ cannot be proper (see Remark \ref{all_leaves_one_vertex}), so we must have $|X|\geq 2$. 

In summary, we know so far that $N$ has at least four vertices, at least two of which are leaves. First, suppose $N$ has exactly four vertices, \emph{two} of which are leaves. This implies that there are two interior vertices, each of them having degree 2. This means that $N$ is a path. In particular, $N$ is not binary and thus not a binary network. Now, suppose $N$ has exactly four vertices, \emph{three} of which are leaves. Then $N$ is a tree on 3 leaves. In particular, $N$ is tree-based, which is a contradiction. Thus, as the number of vertices in an unrooted binary phylogenetic network is even, we can conclude that $|V| \geq 6$.

Now assume that $N$ has strictly fewer than 12 vertices in total. As $N$ is binary and as we have already seen that binary networks have an even number of vertices, this means that $N$ has at most 10 vertices in total. We now distinguish two main cases, which can both be subdivided into two subcases:

\begin{itemize}
\item First assume $|X|=2$. Without loss of generality, let $X=\{x,y\}$ and let $u$ and $v$ denote the vertices adjacent to $x$ and $y$, respectively. Now, consider the leaf connecting procedure. As $|X|=2$, $\mathcal{LCON}(N)$ consists of precisely one element, which we denote by $G(N)$. Note that $G(N)$ is a cubic graph as $N$ is binary. Moreover, recall that the construction of $G(N)$ might have required the introduction of new vertices, say $a$ and $b$, to avoid parallel edges. We now distinguish between two cases.

\begin{itemize}\item If in the construction of $G(N)$ no vertices $a$ and $b$ had to be added, then $G(N)$ contains precisely $|V|-2$ vertices. As $N$ by assumption contains fewer than 12 vertices, $G(N)$ contains fewer than 10 vertices. Thus, as $G(N)$ is cubic and therefore contains an even number of vertices, $G(N)$ contains at most 8 vertices. However, by Lemma \ref{cubic_hamiltonian}, we know that up to 8 vertices there is a Hamiltonian path from $u'$ to $v'$ for every edge $\{u',v'\}$ in a cubic graph. So in particular, $G(N)$ contains a Hamiltonian path from $u$ to $v$, i.e. from the attachment point of the first leaf to the attachment point of the second leaf. Adding edge $\{u,v\}$ to this path yields a Hamiltonian cycle using this new edge, which was not contained in $N$, and thus, by Proposition \ref{cubic_tree-based2}, $N$ is tree-based, which is a contradiction. 
\item If in the construction of $G(N)$ vertices $a$ and $b$ had to be added, then $G(N)$ contains precisely $|V|$ vertices ($x$ and $y$ have been deleted, but $a$ and $b$ have been added). As $N$ by assumption contains fewer than 12 vertices, $G(N)$ also contains fewer than 12 vertices. As the number of vertices in any cubic graph is even, $G(N)$ contains at most 10 vertices. As above, if $G(N)$ contains at most 8 vertices, there is a Hamiltonian path from $u$ to $a$, which can be extended to a Hamiltonian cycle by adding edge $\{a,u\}$, and thus, by Proposition \ref{cubic_tree-based2}, $N$ is tree-based, which is a contradiction. 
So let us consider the case where $G(N)$ contains precisely 10 vertices. Note that there are only 19 different cubic graphs with 10 vertices, and only two of them are not Hamiltonian (cf. \citet{Bussemaker1976}). Now if $G(N)$ is one of the Hamiltonian graphs, it contains a Hamiltonian cycle. We now argue that each such cycle must use edge $\{a,b\}$. Note that by construction of $G(N)$ through leaf connection $a$ is only adjacent to $u$, $v$ and $b$, and $b$ only to $u$, $v$ and $a$. Therefore, the Hamiltonian cycle will connect $u$ and $v$ in two ways, namely with one path visiting all vertices except for $a$ and $b$, and additionally with a path only visiting $a$ and $b$. For instance, the paths $u,a,b,v$ or $u,b,a,v$ would be possible. In all such cases, edge $\{a,b\}$ is necessarily contained in all Hamiltonian cycles. So deleting edge $\{a,b\}$ leads to a Hamiltonian path from $a$ to $b$. Subsequently, suppressing $a$ and deleting $b$ as well the edges $\{u,b\}$ and $\{v,b\}$, leads to a Hamiltonian path from $u$ to $v$. Using the same arguments as above, this implies that $N$ is tree-based, which would be a contradiction. 
So $G(N)$ has to be one of the two non-Hamiltonian cubic graphs with 10 vertices. These two graphs are depicted in Figure \ref{Fig_Cubic10}.

\begin{figure}[htbp]
	\centering
	\includegraphics[scale=0.275]{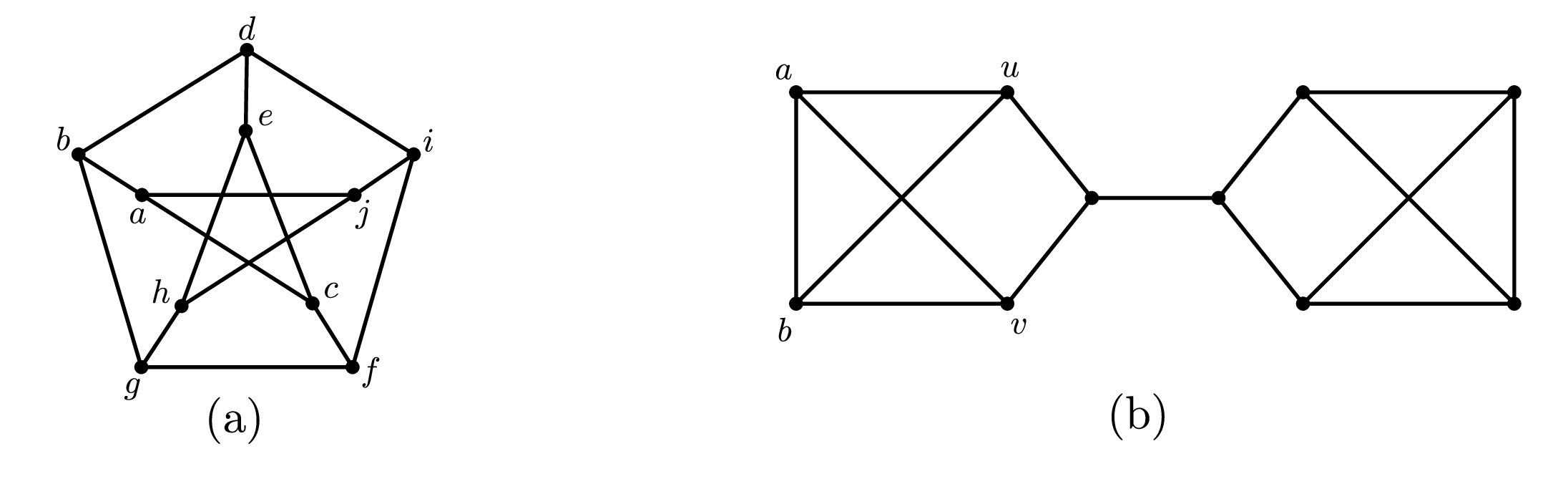}
	\caption{There are only two non-Hamiltonian cubic graphs with 10 vertices, where (a) is the so-called Petersen graph. The vertex labels in (a) refer to the labels as needed in the proof of Theorem \ref{level4_theorem}. In particular, if we connect vertices $a$ and $j$ and delete leaves $x$ and $y$ in the second graph of Figure \ref{Fig_correctExample}, the resulting graph is the Petersen graph as labeled here. }
	\label{Fig_Cubic10}
\end{figure}

Note that the first one, namely the Petersen graph, contains no pair $a$ and $b$ of vertices such that $a$ and $b$ are adjacent to one another and to the same two other vertices $u$ and $v$. So as we assume that such vertices $a$ and $b$ were added during the construction of $G(N)$, $G(N)$ cannot be the Petersen graph. The other one of these two graphs, which is depicted in Figure \ref{Fig_Cubic10} (b),  has two possible positions for the pair $a$ and $b$, but it also has a cut edge, and the positions for $a$ and $b$ (one of which is depicted in Figure \ref{Fig_Cubic10} (b)) are such that both $a$ and $b$ would be at the same side of the cut edge. Thus, by the construction procedure leading to $G(N)$, also $x$ and $y$ must have been on the same side of the cut edge, but then $N$ cannot be proper. This is a contradiction.   
\end{itemize}
In summary, if $|X|=2$, $N$ must contain at least 12 vertices.

\item Now assume $|X|>2$. If necessary, we first perform the pre-processing step of the leaf connecting procedure introduced in Section \ref{sec_preliminaries} and denote the resulting reduced taxon set of $N$ by $X^r$. Note that $\mathcal{LCON}(N)$ possibly contains more than one graph. So let $G(N)$ in the following be an arbitrary element of $\mathcal{LCON}(N)$. We now distinguish two cases.

\begin{itemize}
\item Suppose $|X^r|$ is even. Consider $G(N)$ and $\widetilde{N}$, where $\widetilde{N}$ is the second to last graph in the construction of $G(N)$ according to the definition of $\mathcal{LCON}(N)$. In particular, $\widetilde{N}$ is the graph which we get when only two leaves are left, which we would have to connect in order to receive the final graph $G(N)$. Thus, by construction $\widetilde{N}$ has at most as many vertices as $N$, so by assumption fewer than 12, and it has two leaves. But as $N$ is not tree-based by assumption, neither is $\widetilde{N}$, because any support tree for $\widetilde{N}$ would lead to a support tree of $N$. To see this, we distinguish between two cases:
\begin{itemize}
\item If the support tree of $\widetilde{N}$ only contains edges that are both present in $\widetilde{N}$ and in $N$, we can obtain a support tree for $N$ by re-attaching the leaves at their former positions. 
\item Otherwise, suppose that the support tree of $\widetilde{N}$ contains at least one edge that is not present in $N$. Note that there are two potential types of edges that can be present in $\widetilde{N}$ but not in $N$:
	\begin{itemize}
	\item Edges that were introduced to avoid multiple edges between two attachment points, say $u$ and $v$, of leaves, i.e. the edges $\{u,a\}, \, \{u,b\}, \, \{a,v\}, \, \{b,v\}$ and $\{a,b\}$. If the support tree of $\widetilde{N}$ uses any of these edges, we can obtain a support tree for $N$ as follows: Delete these edges as well as vertices $a$ and $b$ from the support tree. Additionally, add a new edge $e=\{u,v\}$ to the support tree of $\widetilde{N}$ (note that this is allowed as $e=\{u,v\}$ must have been contained in $N$, which led to the introduction of $a$ and $b$) and re-attach the leaves incident to $u$ and $v$ in $N$ to $u$ and $v$, respectively. We can repeat this procedure for all edges of this type.
	\item Edges $e=\{u,v\}$ between two attachment points ($u$ and $v$) of leaves. If the support tree of $\widetilde{N}$ uses such an edge, we can obtain a support tree for $N$ as follows: First of all, delete the edge $\{u,v\}$ from the support tree of $\widetilde{N}$. Note that this disconnects the support tree. Let $T_1$ and $T_2$ denote its two connected components and assume that $u$ is in $T_1$ and $v$ is in $T_2$. Moreover, note that both $u$ and $v$ are of degree 3 by construction, thus $T_1$ and $T_2$ cannot be single vertices. 
	As $N$ is a connected graph, in particular there exists a vertex $u' \neq u$ in $T_1$ and a vertex $v' \neq v$ in $T_2$ such that $u'$ and $v'$ are connected by an edge $e'=\{u',v'\}$ in $N$ (if $T_1$ and $T_2$ were only connected via $\{u,v\}$, $N$ would not have been connected, as $\{u,v\}$ is only present in $\widetilde{N}$ and not in $N$). We now add $e'=\{u',v'\}$ to the support tree of $\widetilde{N}$, and re-attach the leaves incident to $u$ and $v$ in $N$ to $u$ and $v$, respectively. Again, we can repeat this procedure for all edges of this type.
	\end{itemize}
\end{itemize}
In all cases, we can construct a support tree for $N$ from a support tree of $\widetilde{N}$. This is a contradiction as $N$ is not tree-based. Thus, $\widetilde{N}$ has to be non-tree-based.
However, then $\widetilde{N}$ would be a non-tree-based network with two leaves and strictly fewer than 10 interior vertices, which contradicts the first part of the proof.
%%%
\item Suppose $|X^r|$ is odd. By construction, as we assume $N$ has at most 10 vertices, also $G(N)$ can have at most 10 vertices. This is due to the fact that in each step during the construction of $G(N)$, either two leaves are deleted or, in the last step, one leaf and its attachment point are deleted. So in all steps, two vertices are deleted and at most two new vertices are added (if parallel edges need to be avoided), so the total number of vertices cannot increase.

Now, if the resulting graph $G(N)$ has at most eight vertices, we already know by Lemma \ref{cubic_hamiltonian} that for each edge $e=\{u',v'\}$ it contains a Hamiltonian cycle from $u'$ to $v'$. Thus, using the same arguments as in the case where $|X|=2$ combined with Lemma \ref{N-x}, $N$ must have been tree-based, which is a contradiction\footnote{For instance, we can apply these arguments to the case where, on the way to constructing $G(N)$, the resulting network has 3 leaves and the last pair gets connected, before we deal with the last singleton leaf. At least one of the edges resulting from connecting this last pair of leaves must be contained in any Hamiltonian cycle and thus leads to a Hamiltonian path when we disregard it. }. If, on the other hand, $G$ has precisely 10 vertices, then again, as in the case where $|X|=2$, we only need to consider the two cubic graphs with 10 vertices which are non-Hamiltonian depicted in Figure \ref{Fig_Cubic10}. The Petersen graph as before does not have any pair of vertices $a$ and $b$ that could be suppressed so that we have parallel edges, which implies that no vertices have been added during the deletion of leaves. This means that all 10 vertices of the Petersen graph were already there in $N$, plus at least three leaves. So in total, $N$ would have at least 13 vertices, which is a contradiction to the assumption that $N$ has fewer than 12 vertices.

The other non-Hamiltonian cubic graph with 10 vertices, however, namely the one depicted in Figure \ref{Fig_Cubic10} (b) has the property that wherever we attach at least two leaves, the network is immediately tree-based whenever it is proper: If the leaf set is such that it is distributed at both sides of the cut edge, the network is tree-based, and if all leaves are on the same side of the cut edge, the network is not proper (whether or not we delete the candidate vertices $a$ and $b$ as depicted in Figure \ref{Fig_Cubic10} (b), which may have been added during the deletion of the leaves, does not matter). Both scenarios contradict the assumption that the network is proper but not tree-based. 
\end{itemize}
\end{itemize}

So in all cases, the result is a contradiction, so every network with fewer than 12 vertices is tree-based. This completes the proof.
\end{proof}

We are now in the position to prove Theorem \ref{level4_theorem}. The first part of the proof is identical to the proof presented in \citet{Francis2018}, but in the second part we use a different argument, in particular, we do not use so-called level-$k$ generators.

\begin{proof}[Proof (Theorem \ref{level4_theorem})]
We have to show that all proper level-0,1,2,3 and 4 networks are tree-based.
The first part of the proof is analogous to the first part of the proof of Theorem \ref{level_non-binary}, i.e. we reduce the analysis to simple networks $B_{N}$ with exactly two leaves and show that all proper binary level-0,1,2,3, and 4 networks are tree-based.
As in the non-binary case, it is immediately clear that any level-$k$ network is tree-based if $k=0$ or $k=1$:  
For $k=0$ the network is a tree and is thus tree-based. For $k=1$, at most one edge has to be removed from each non-trivial blob to obtain a tree; 
this tree is a support tree, because removing at most one edge from each non-trivial blob cannot induce any new leaves (same argument as in the non-binary case). 
Thus, let us now consider $k \geq 2$.
Due to Lemma \ref{12vertices} we know that any minimal proper non-tree-based network has at least 12 vertices, 2 of which are leaves. This implies that any $B_{N}$ that is not tree-based has at least 12 vertices.\footnote{Note that by Lemma \ref{binarycutvertex}, as blobs do not contain cut edges and as we are in the binary case, there can also be no cut vertices (as these are always incident to cut edges in this case). So $B_{N}$ must be proper, which indeed justifies the usage of Lemma \ref{12vertices}. }
Additionally, by Lemma \ref{blob_nodes} we know that the number of vertices in a non-trivial blob incident to two cut edges (and thus corresponding to a simple network $B_{N}$ with two leaves) in a level-$k$ network is bounded from above by $2k$.  
Now, as we assume that all $B_{N}$ have exactly two leaves, any $B_{N}$ that is not tree-based has to correspond to a non-trivial blob $B$ with at least 10 vertices. Thus,
$$ 10 \leq n \leq 2k,$$ where $n$ denotes the number of vertices of $B_N$.
This immediately implies that $k \geq 5$, thus there cannot be an unrooted binary non-tree-based level-$k$ network with $k \leq 4$. 

To prove the last statement of the theorem, consider either one of the two level-5 networks depicted in Figure \ref{Fig_correctExample}. These networks can be seen to not be tree-based as follows.
If they were tree-based, then there would be a path from $x$ to $y$ visiting every vertex exactly once. Any such path must begin with the edge $\{x,a\}$ and end with the edge $\{j,y\}$. Now, for the network at the top of Figure \ref{Fig_correctExample}, it is straightforward to see that every path visiting both the vertices at the top (i.e. vertices $b,c,d,e$) and the vertices at the bottom (i.e. vertices $f,g,h,i$) must visit either vertex $a$ or $j$ twice, which is a contradiction. For the second network in Figure \ref{Fig_correctExample} this is a bit harder to see, but it can be verified as follows: Connect the two leaves to construct $\mathcal{LCON}(N)$, which in this case (as there are only two leaves) contains precisely one graph, say $G$. $G$ is in fact isomorphic to the Petersen graph, cf. Figure \ref{Fig_Cubic10}(a). In particular, $G$ is not Hamiltonian. Thus, there is also no Hamiltonian cycle using the edge $\{a,j\}$ in $G$, which in turn implies that there cannot be a Hamiltonian path from $a$ to $j$ in $G$. Thus, there can also not be a support tree $T$ for $N$ -- otherwise, removing $x$ and $y$ would induce a Hamiltonian path from $a$ to $j$ along $T$ in $G$. This completes the proof. 
\end{proof}

\setcounter{proposition}{2}

\begin{proposition}
There are precisely two minimal (in the number of vertices) proper unrooted binary phylogenetic level-5 networks that are not tree-based. Both of them contain exactly 12 vertices. In particular, all other proper unrooted binary phylogenetic level-5 networks that are not tree-based contain more than 12 vertices.
\end{proposition}

\begin{proof}
By Lemma \ref{12vertices}, we know that 12 vertices are required for a network to be non-tree-based. Thus, the two proper unrooted binary non-tree-based phylogenetic level-5 networks depicted in Figure \ref{Fig_correctExample} are minimal. It remains to show that these two networks are the only non-tree-based level-5 networks with 12 vertices.
We verified this by an exhaustive search with Mathematica \citet{Mathematica}, which was conducted in the following way: 
First, we obtained a list of all connected simple graphs with 10 vertices (11716571 in total) from the \enquote{House of Graphs} database (cf. \citet{Brinkmann2013}). These were then analyzed for potential binary networks with 10 interior vertices and 2 leaves by checking whether they contained exactly 8 vertices of degree 3 and 2 vertices of degree 2 (which we called $u$ and $v$ and to which we subsequently attached leaves) using the Mathematica function \texttt{VertexDegree[$\cdot$]}. The resulting 113 graphs were analyzed for tree-basedness in the following way: 
	\begin{itemize}
	\item We attached one leaf to each of the two degree-2 vertices $u$ and $v$.
	\item The two leaves were then connected according to the leaf connecting procedure (note that as there are only 2 leaves,  $\mathcal{LCON}(N)$ contains only one graph).
	\item We then used the Mathematica function \texttt{FindHamiltonianCycle[$\mathcal{LCON}(N)$, All]} to find all Hamiltonian cycles of $\mathcal{LCON}(N)$. We then checked whether one of them used at least one edge in $E(\mathcal{LCON}(N)) \setminus E(N)$.  If so, this Hamiltonian cycle corresponds to a Hamiltonian path from $u$ to $v$ in $\mathcal{LCON}(N)$, meaning that $N$ is tree-based (cf. Proposition \ref{cubic_tree-based2}). 
	\end{itemize}
This left us with 10 non-tree-based networks, which were then filtered for proper networks. It turned out that 8 of them were not proper, i.e. there are exactly 2 proper binary phylogenetic networks with 12 vertices, both of which are level-5 networks. They are the ones depicted in Figure \ref{Fig_correctExample}.

Furthermore, both networks in Figure \ref{Fig_correctExample} are symmetric in the sense that swapping their leaves $x$ and $y$ leads to an isomorphic network, respectively. So, there are indeed only two minimal proper unrooted binary \emph{phylogenetic} networks that are not tree-based. More precisely, both non-leaf-labeled graphs obtained by the exhaustive search described above correspond to precisely one leaf-labeled network each.
\end{proof}

\subsection{Supplementary figures}
\begin{figure}[H]
\centering
\includegraphics[scale=0.5]{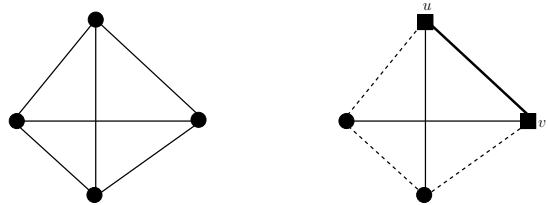}
\caption{Left: There exists exactly one cubic graph with 4 vertices. Right: For $e=\{u,v\}$ there exits a Hamiltonian path from $u$ to $v$ indicated by dashed lines.}
\label{figcubic1}
\end{figure}
\begin{figure}[H]
\centering
\includegraphics[scale=0.5]{cubic2.eps}
\caption{All cubic graphs with 6 vertices.}
\label{figcubic2}
\end{figure}
\begin{figure}[H]
\centering
\includegraphics[scale=0.35]{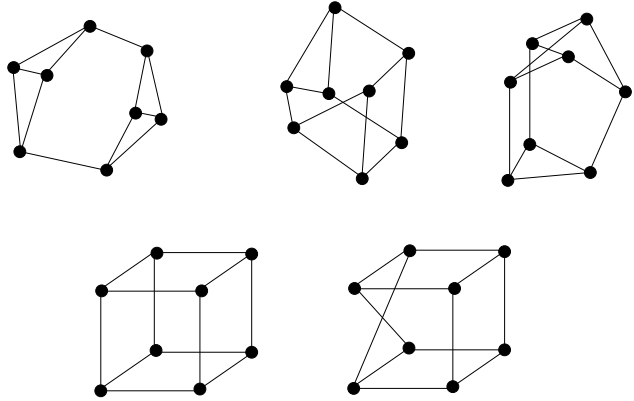}
\caption{All cubic graphs with 8 vertices.}
\label{figcubic3}
\end{figure}

\subsubsection{Catalog of all proper tree-based networks with up to 7 vertices}
In the following all proper tree-based networks with up to 7 vertices are depicted. For each network, a support tree is shown in bold lines and the additional network edges are given by dashed lines.

\begin{figure}[H]
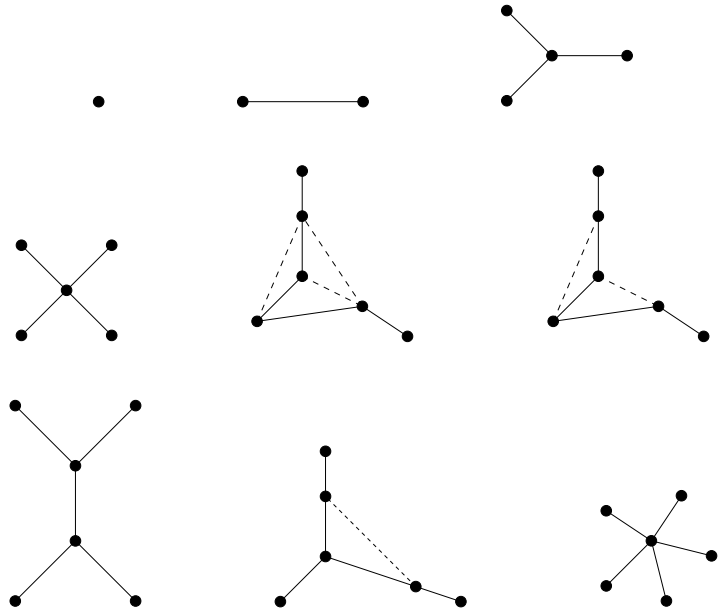

\centering
\includegraphics[scale=0.4]{catalog1.eps} \hspace{15mm}
\includegraphics[scale=0.4]{catalog2.eps} \hspace{15mm}
\includegraphics[scale=0.4]{catalog3.eps} \\ \vspace{7mm}
\includegraphics[scale=0.4]{catalog4.eps} \hspace{15mm}
\includegraphics[scale=0.4]{catalog5.eps} \hspace{15mm}
\includegraphics[scale=0.4]{catalog6.eps} \\ \vspace{7mm}
\includegraphics[scale=0.4]{catalog7.eps} \hspace{15mm}
\includegraphics[scale=0.4]{catalog8.eps} \hspace{15mm}
\includegraphics[scale=0.4]{catalog9.eps} \\ \vspace{7mm}
\caption{All proper tree-based networks with up to 6 vertices (up to leaf labeling).}
\label{Catalog1}
\end{figure}

%%%
\begin{figure}[H]
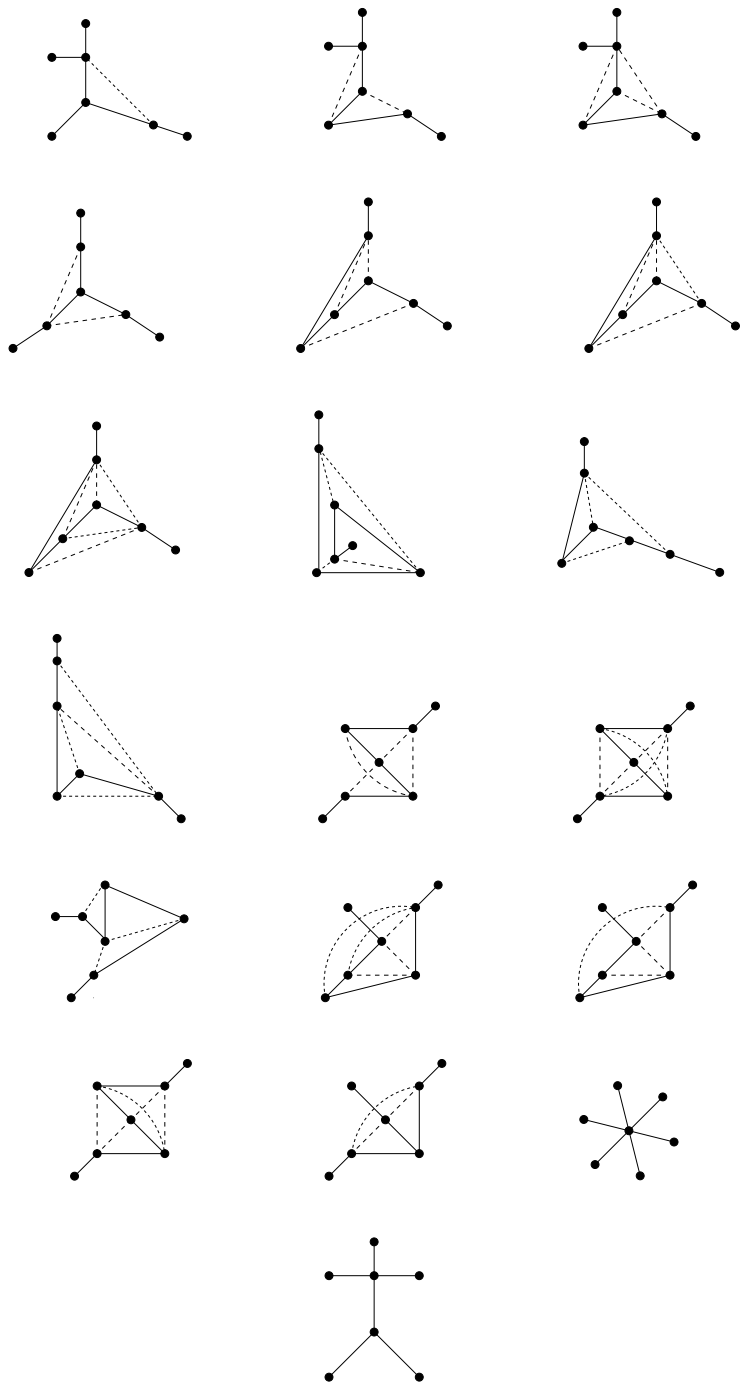

\centering
\includegraphics[scale=0.3]{catalog10.eps} \hspace{15mm}
\includegraphics[scale=0.3]{catalog11.eps} \hspace{15mm}
\includegraphics[scale=0.3]{catalog12.eps} \\ \vspace{7mm}
\includegraphics[scale=0.3]{catalog13.eps} \hspace{15mm}
\includegraphics[scale=0.3]{catalog14.eps} \hspace{15mm}
\includegraphics[scale=0.3]{catalog15.eps} \\ \vspace{7mm}
\includegraphics[scale=0.3]{catalog16.eps} \hspace{15mm}
\includegraphics[scale=0.3]{catalog17.eps} \hspace{15mm}
\includegraphics[scale=0.3]{catalog18.eps} \\ \vspace{7mm}
\includegraphics[scale=0.3]{catalog19.eps} \hspace{15mm}
\includegraphics[scale=0.3]{catalog20.eps} \hspace{15mm}
\includegraphics[scale=0.3]{catalog21.eps} \\ \vspace{7mm}
\includegraphics[scale=0.3]{catalog22.eps} \hspace{15mm}
\includegraphics[scale=0.3]{catalog23.eps} \hspace{15mm}
\includegraphics[scale=0.3]{catalog24.eps} \\ \vspace{7mm}
\includegraphics[scale=0.3]{catalog25.eps} \hspace{15mm}
\includegraphics[scale=0.3]{catalog26.eps} \hspace{15mm}
\includegraphics[scale=0.3]{catalog27.eps} \\ \vspace{7mm}
\includegraphics[scale=0.3]{catalog28.eps}
\caption{All proper tree-based networks with 7 vertices (up to leaf labeling).}
\label{Catalog2}
\end{figure}

\subsubsection{Catalog of all proper non-tree-based networks with 6 interior vertices and 2 leaves}
\begin{figure}[htbp]
\centering
\includegraphics[scale=0.2]{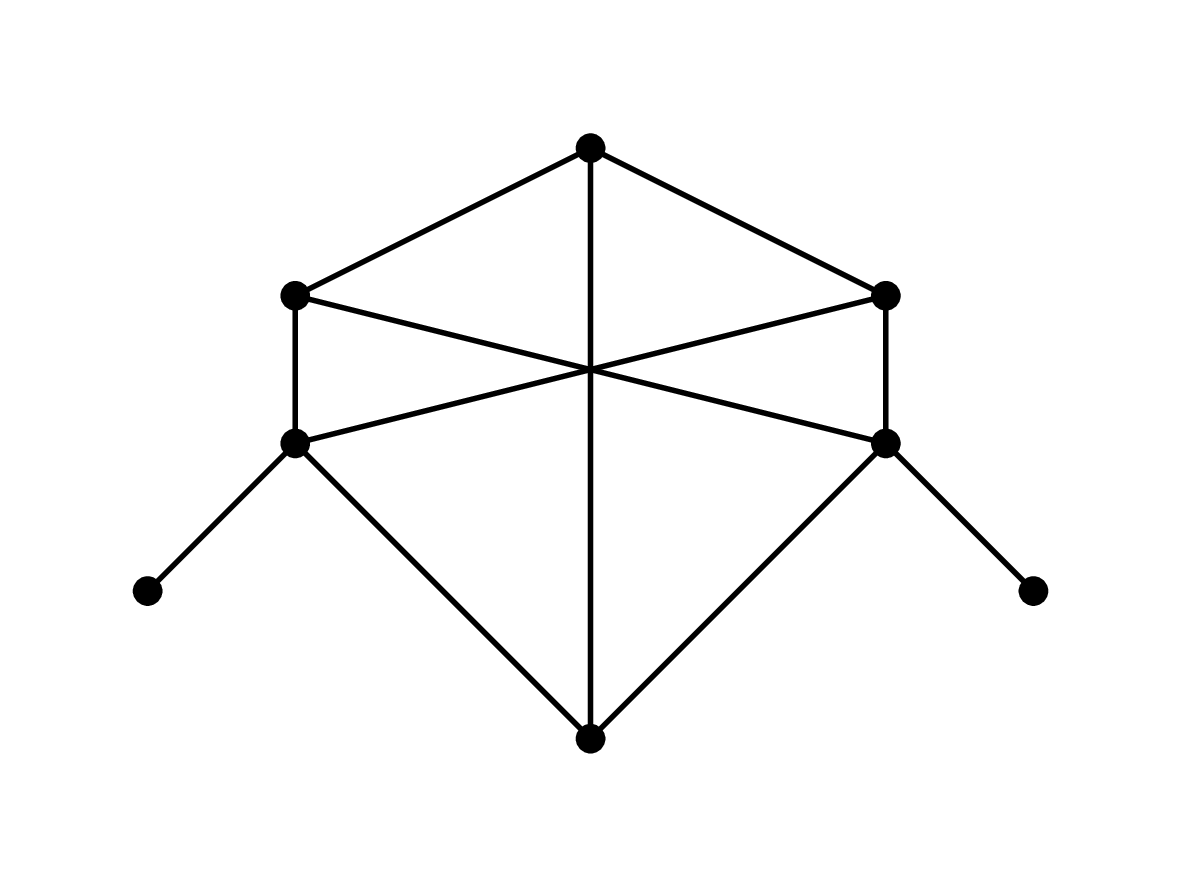}\hspace{15mm}
\includegraphics[scale=0.2]{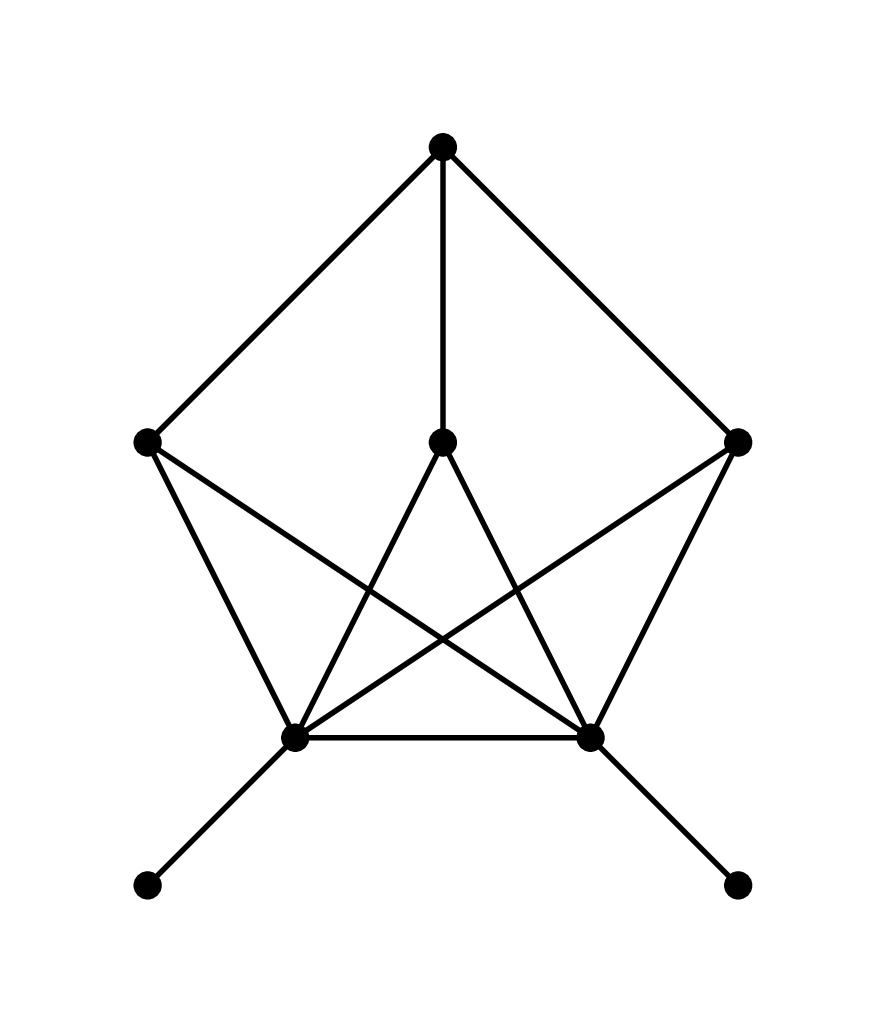}
\includegraphics[scale=0.2]{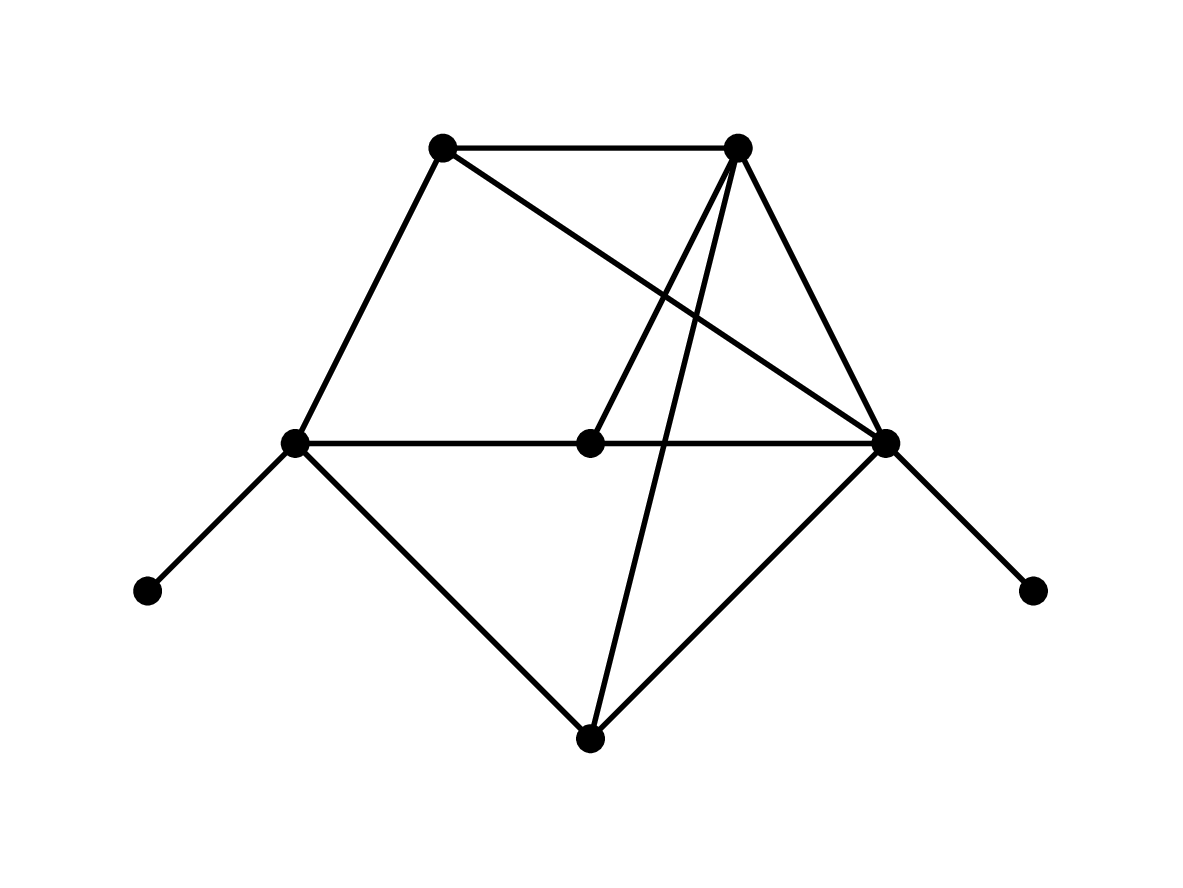}\hspace{15mm}
\includegraphics[scale=0.2]{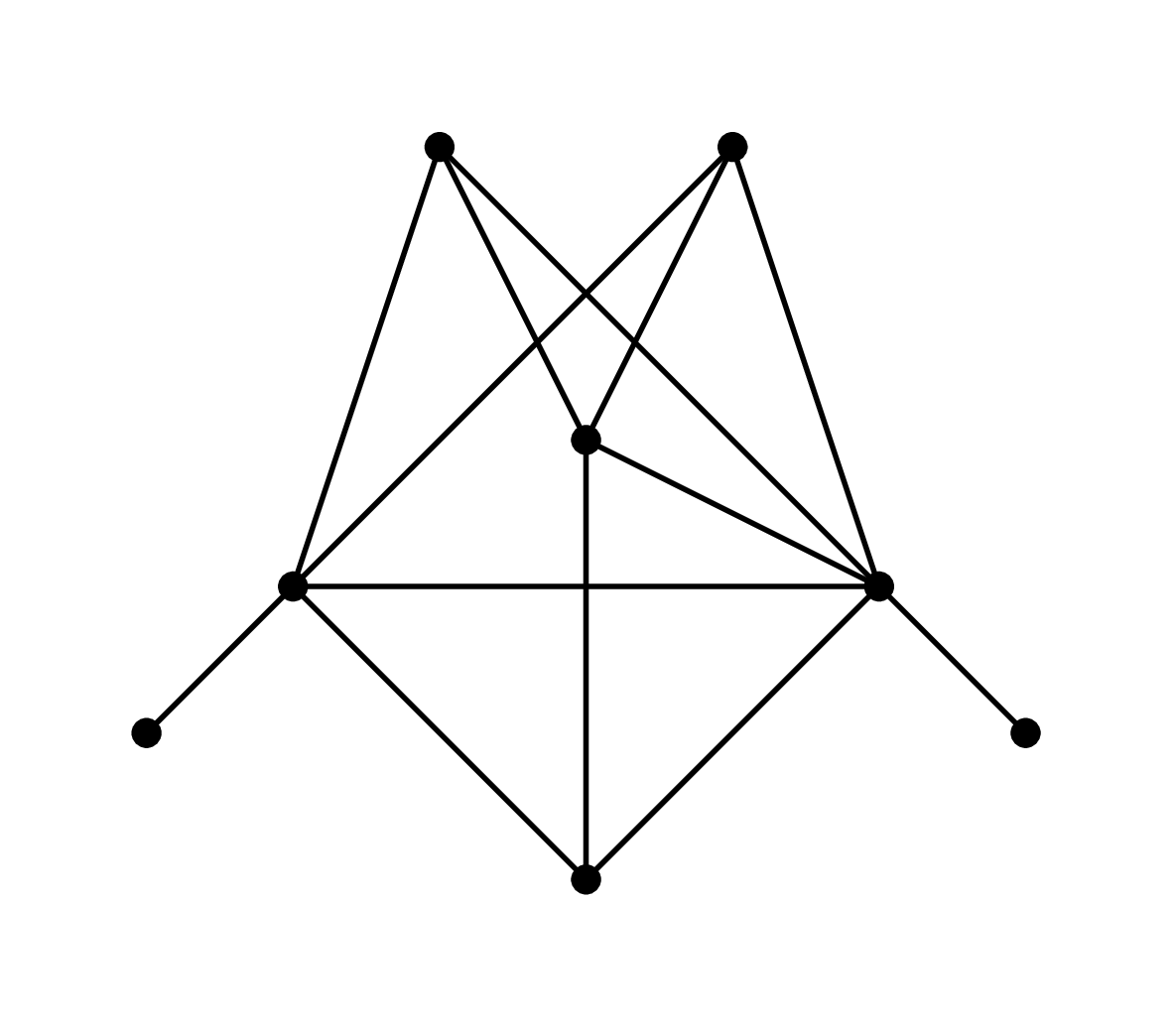}
\includegraphics[scale=0.2]{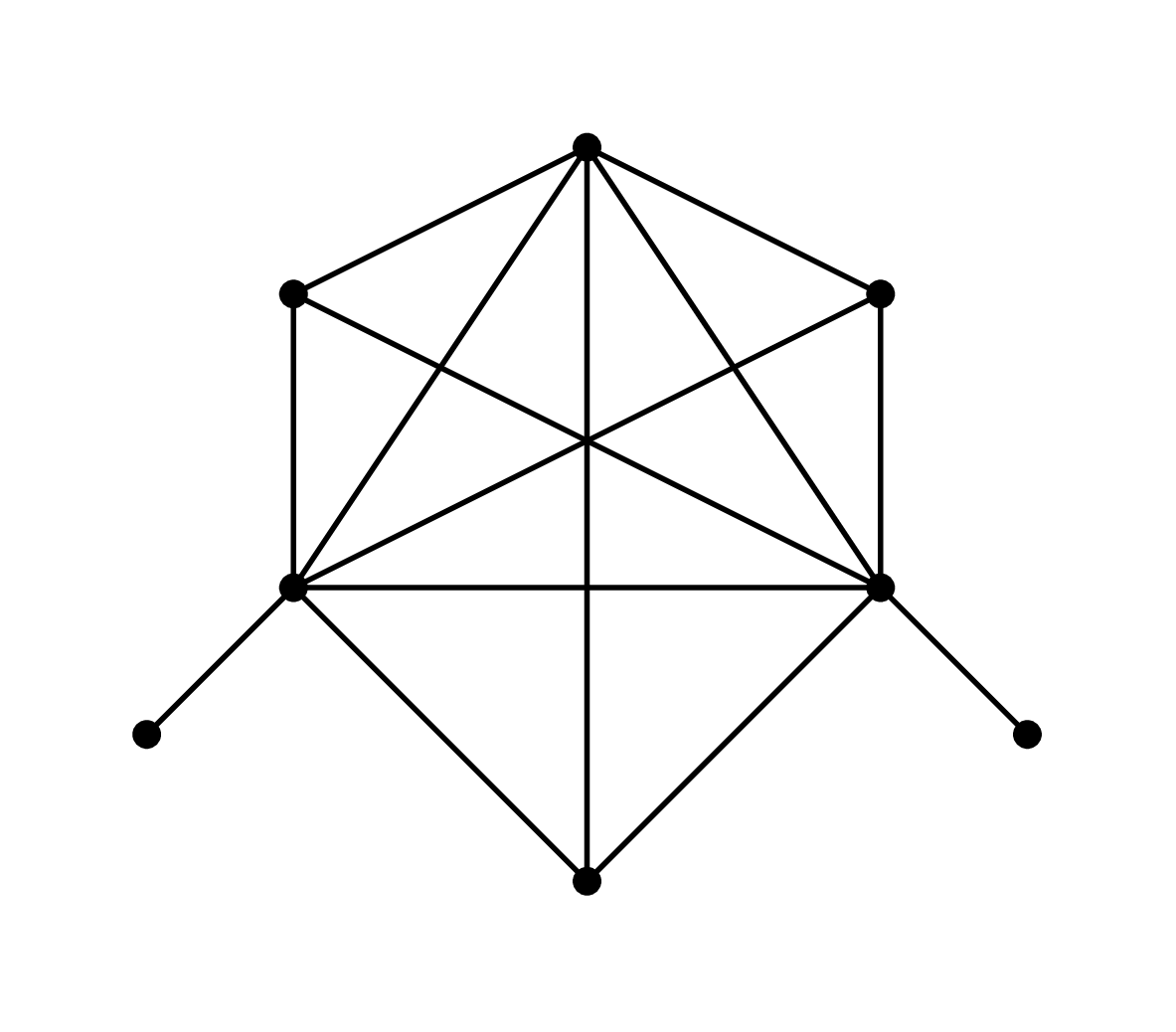}\hspace{15mm}
\includegraphics[scale=0.2]{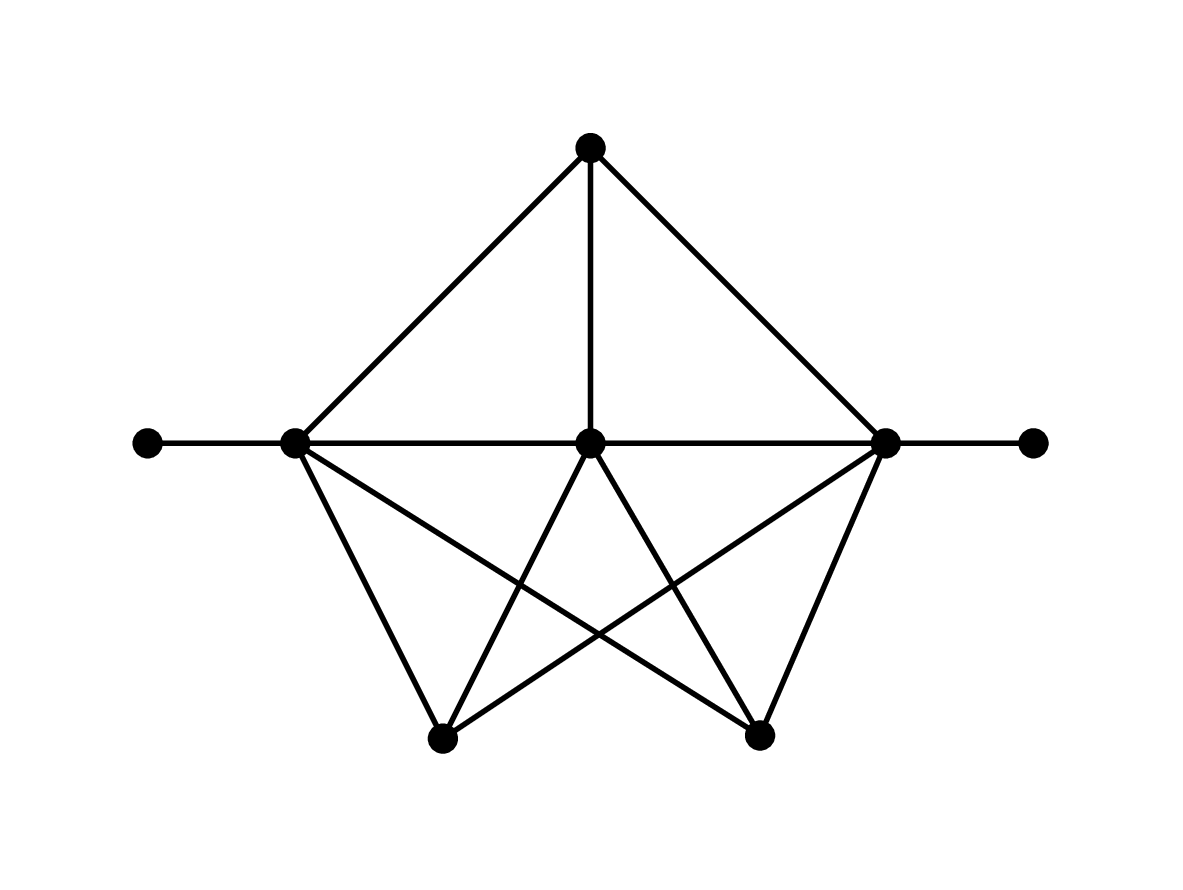}
\includegraphics[scale=0.2]{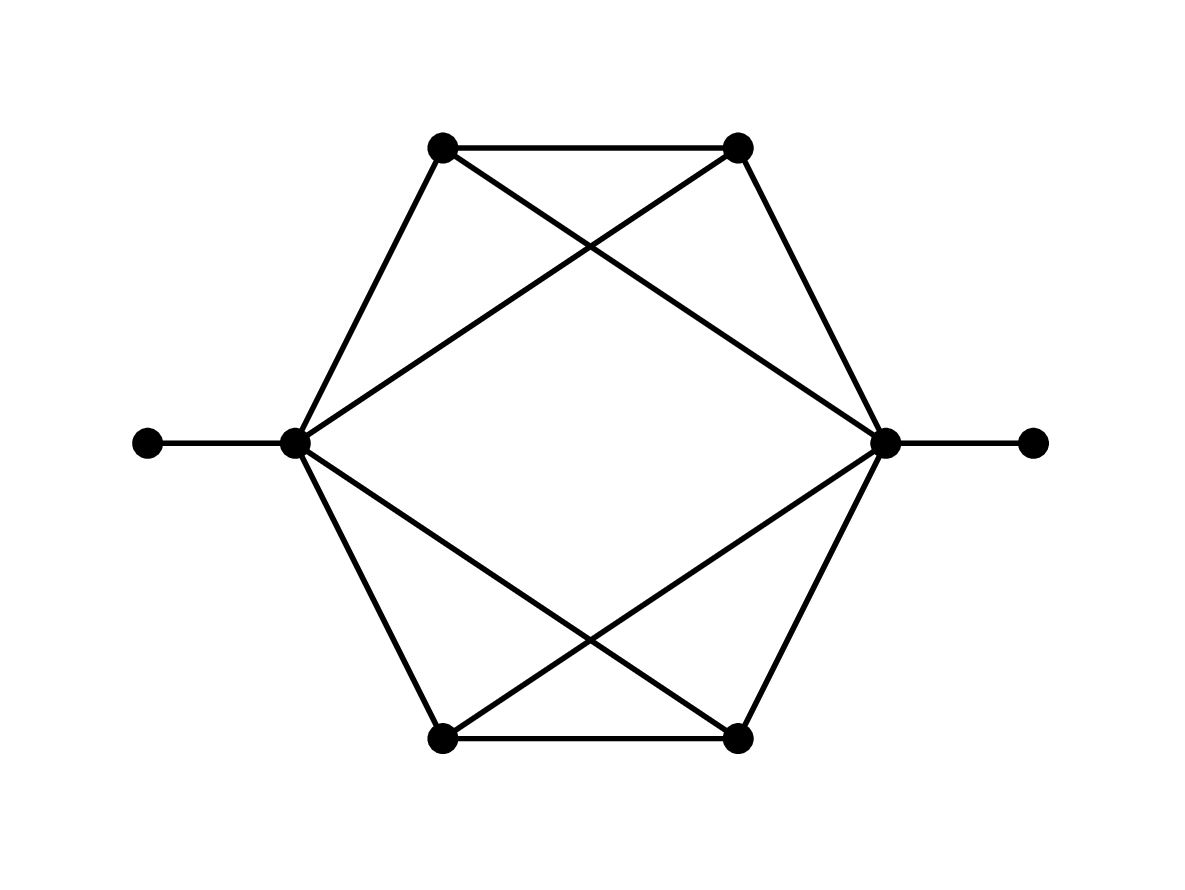}\hspace{15mm}
\includegraphics[scale=0.2]{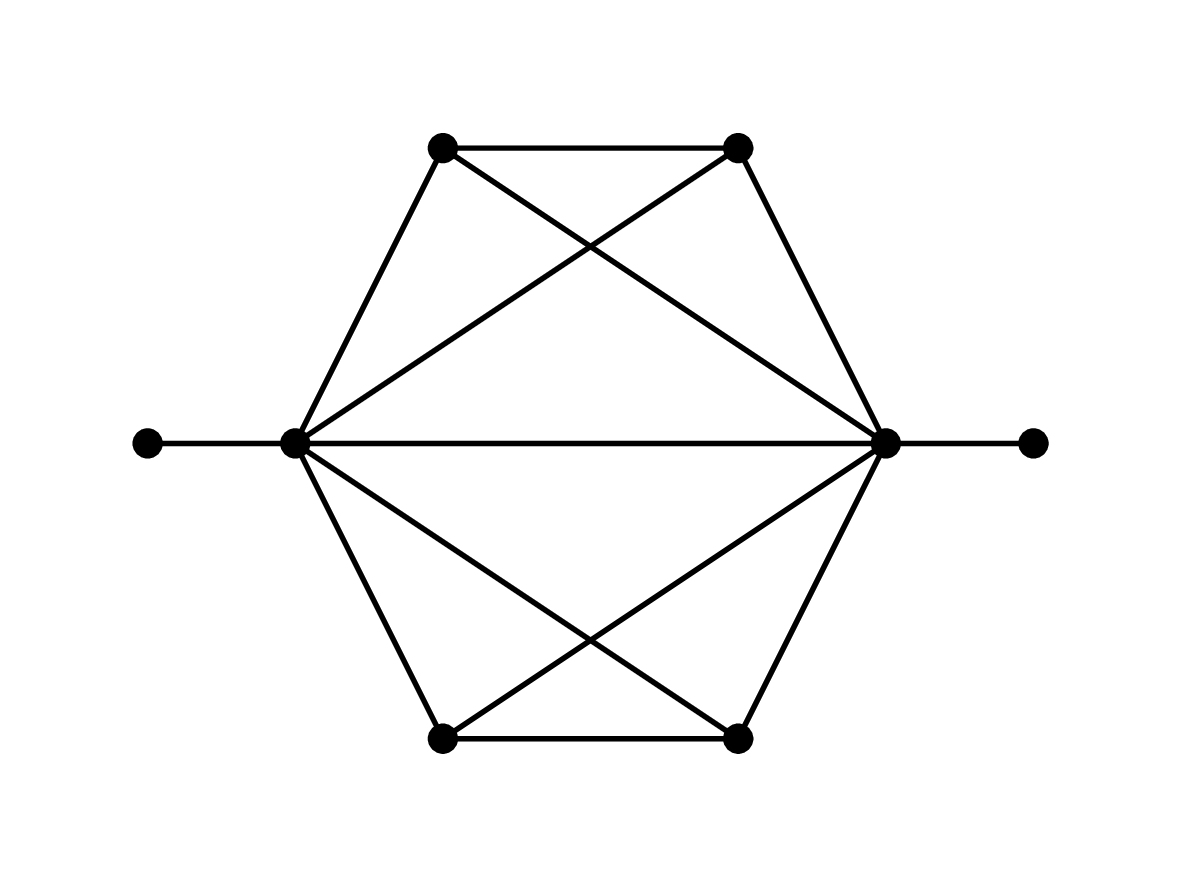}
\caption{All proper unrooted non-tree-based networks on 8 vertices (up to leaf labeling), 2 of which are leaves.}
\label{Fig_Catalog8}
\end{figure}

\end{document}